\documentclass[runningheads]{llncs}
\usepackage{amsmath}
\usepackage{graphicx}
\usepackage{censor}
\usepackage{cite}
\usepackage{color}
\usepackage{hyperref}
\usepackage{outlines}
\usepackage{amsfonts}
\usepackage{pifont}
\usepackage{wasysym}
\usepackage{amssymb}
\usepackage{todonotes}
\usepackage{soul}
\usepackage{wrapfig}
\usepackage{booktabs}
\usepackage{float}
\usepackage{array}
\usepackage{makecell}
\usepackage{paralist} 
\usepackage[T1]{fontenc}
\usepackage{tikz}
\usetikzlibrary{arrows, calc, shapes,petri}
\usepackage{extarrows}
\usepackage{comment}

\newcommand{\objects}{\mathcal O}
\newcommand{\objtypes}{\Sigma}
\newcommand{\dom}{dom}
\newcommand{\codom}{codom}
\newcommand{\tup}[1]{\langle{#1}\rangle}
\newcommand{\allvars}{\mathcal X}
\newcommand{\varset}{X}
\newcommand{\nuvarset}{\Upsilon}
\newcommand{\listvarset}{\varset_{\mathit{list}}}
\newcommand{\vartype}{\textup{type}}
\newcommand{\activities}{\mathcal A}

\newcommand{\inflow}{F_{in}}
\newcommand{\outflow}{F_{out}}
\newcommand{\pre}[1]{\bullet{#1}}
\newcommand{\post}[1]{{#1}\bullet}
\newcommand{\invars}[1]{\textup{vars}_{in}(#1)}
\newcommand{\outvars}[1]{\textup{vars}_{out}(#1)}
\newcommand{\listtype}[1]{[{#1}]}
\newcommand{\colors}{\mathcal C}
\newcommand{\coloring}{\textup{color}}
\newcommand{\vars}[1]{\textup{vars}(#1)}
\newcommand{\tokens}{\vec\objects}

\newcommand{\goto}[1]{\mathrel{\raisebox{-1pt}{$\xrightarrow{#1}$}}}
\newcommand{\m}[1]{\mathsf{#1}}
\newcommand{\Minit}{M_{\mathit{init}}}
\newcommand{\Mfinal}{M_{\mathit{final}}}
\newcommand{\Mempty}{M_\emptyset}
\newcommand{\Cons}{\mathit{cons}}
\newcommand{\Prod}{\mathit{prod}}

\newcommand{\tlink}[1]{t_{#1}} 

\newcommand{\directmap}{\mathbb T_1}
\newcommand{\linkmap}{\mathbb T_R}

\tikzstyle{class}=[draw, rectangle, inner sep=1.5pt, line width=.7pt, scale=1, minimum width=24mm, minimum height=8mm]

\tikzstyle{place}=[draw, circle, inner sep=1.5pt, line width=.7pt, scale=.8, minimum width=6mm]
\tikzstyle{trans}=[draw, rectangle, inner sep=1.5pt, line width=.7pt, scale=.8, minimum width=6mm, minimum height=6mm, fill=gray!10]
\tikzstyle{silentTrans}=[draw, rectangle, inner sep=1.5pt, line width=.7pt, scale=.8, minimum width=2.5mm, minimum height=6mm, fill=black]
\tikzstyle{arc}=[draw, ->, line width=.5pt]
\tikzstyle{fatarc}=[draw, ->, line width=.5pt, double]
\tikzstyle{fatline}=[draw, -, line width=.5pt, double]
\tikzstyle{action}=[scale=.6]
\tikzstyle{starttoken}=[regular polygon, regular polygon sides=3,minimum width=3mm,fill=black,inner sep=0pt,rotate=30]
\tikzstyle{endtoken}=[regular polygon, regular polygon sides=4,minimum width=3mm,fill=black,inner sep=0pt]
\tikzstyle{pseudostart}=[regular polygon, regular polygon sides=3,minimum width=3mm,draw,inner sep=0pt, rotate=30]
\tikzstyle{pseudoend}=[regular polygon, regular polygon sides=4,minimum width=3mm,inner sep=0pt, draw]

\colorlet{frameColor}{yellow!60!white!40}
\colorlet{wheelColor}{red!60!white!50}
\colorlet{handleColor}{blue!70!white!30}
\colorlet{wheelFrameColor}{orange!50}
\colorlet{handleFrameColor}{green!100!white!30}
\tikzstyle{framePlace}=[place,fill=frameColor]
\tikzstyle{wheelPlace}=[place,fill=wheelColor]
\tikzstyle{handlePlace}=[place,fill=handleColor]
\tikzstyle{wheelFramePlace}=[place,fill=wheelFrameColor]
\tikzstyle{handleFramePlace}=[place,fill=handleFrameColor]
\colorlet{ocolor}{blue!70!green!20}
\colorlet{icolor}{yellow!90!red!20}
\colorlet{oicolor}{green!90!blue!20}
\tikzstyle{oplace}=[place,fill=ocolor]
\tikzstyle{otrans}=[trans,fill=ocolor]
\tikzstyle{iplace}=[place,fill=icolor]
\tikzstyle{itrans}=[trans,fill=icolor]
\tikzstyle{oiplace}=[place,fill=oicolor]
\tikzstyle{idplace}=[place,fill=red!30] 
\tikzstyle{dplace}=[place,fill=magenta!20]             
\tikzstyle{oidplace}=[place,fill=black!20]
\tikzstyle{oitrans}=[trans,fill=oicolor]
\tikzstyle{mixtrans}=[trans,shading = axis,rectangle, left color=ocolor, right color=icolor,shading angle=135]
\tikzstyle{oimixtrans}=[trans,shading = axis,rectangle, left color=ocolor, right color=oicolor,shading angle=135]
\tikzstyle{insc}=[scale=.8]
\tikzstyle{splitme}=[rectangle split, rectangle split horizontal,rectangle split parts=2]
\tikzstyle{log}=[fill=white]
\tikzstyle{model}=[fill=gray!10]

\newcommand{\Pplay}{P_{\tikz{\node[starttoken,minimum width=1.8mm] {};}}}
\newcommand{\Pstop}{P_{\tikz{\node[endtoken,minimum width=1.8mm] {};}}}

\begin{document}

\title{To bind or not to bind? Discovering Stable Relationships in Object-centric Processes 
(Extended Version)
}
\titlerunning{Discovering Stable Relationships in Object-centric Processes}

\author{
Anjo Seidel\inst{1} \and
Sarah Winkler\inst{2} \and
Alessandro Gianola\inst{3} \and
Marco Montali\inst{2} \and
Mathias Weske\inst{1}
}

\authorrunning{A. Seidel et al.}

\institute{
    Hasso Plattner Institute, University of Potsdam, Germany\\
    \email{\{anjo.seidel,mathias.weske\}@hpi.de}\and
    Free University of Bozen-Bolzano, Italy
    \email{\{winkler,montali\}@inf.unibz.it} \and
    INESC-ID/Instituto Superior Técnico, Universidade de Lisboa, Portugal
    \email{alessandro.gianola@tecnico.ulisboa.pt}
}
\sloppy 
\maketitle

\begin{abstract}
Object-centric process mining investigates the intertwined behavior of multiple objects in business processes.
From object-centric event logs, object-centric Petri nets (OCPN) can be discovered to replay the behavior of processes accessing different object types.
Although they indicate how objects flow through the process and co-occur in events, OCPNs remain underspecified about the relationships of objects.
Hence, they are not able to represent synchronization, i.e. executing objects only according to their intended relationships, and fail to identify violating executions.
Existing formal modeling approaches, such as object-centric Petri nets with identifiers (OPID), represent object identities and relationships to synchronize them correctly. However, OPID discovery has not yet been studied.
This paper uses explicit data models to bridge the gap between OCPNs and formal OPIDs.
We identify the implicit assumptions of stable many-to-one relationships in object-centric event logs, which implies synchronization of related objects.
To formally underpin this observation, we combine OCPNs with explicit stable many-to-one relationships in a rigorous mapping from OCPNs to OPIDs explicitly capturing the intended stable relationships and the synchronization of related objects.
We prove that the original OCPNs and the resulting OPIDs coincide for those executions that satisfy the intended relationships.
Moreover, we provide an implementation of the mapping from OCPN to OPID under stable relationships.

\keywords{Business Process Management \and Data Modeling \and Object-centric Petri Nets \and Conformance Checking  \and Synchronization.}
\end{abstract}
\section{Introduction}
\label{sec:intro}

Process mining analyzes the behavior of information systems through event logs to gain insight into the behavior of the system~\cite{vanderaalstProcessMiningManifesto2012a}.
Process mining techniques strive between two contrasting forces when dealing with process models: capturing complex real-life processes demands \emph{expressiveness}, while performance and memory consumption of the adopted algorithms call for \emph{simplicity}.
Balancing these contrasting forces is difficult given that
within a process mining pipeline, the same model can be seen as descriptive, like a model discovered from data, or as prescriptive, e.g., when later on used for conformance checking:
while a descriptive model can favor simplicity over precision, a prescriptive one should be expressive enough to precisely capture the intended behavior.
Achieving a suitable trade-off appears particularly difficult within the novel research area of object-centric process science, where processes co-evolve multiple, interrelated objects at once~\cite{vanderaalstObjectCentricProcessMining2019}.
Executions of such processes are typically represented in
\emph{object-centric event logs} (OCELs)~\cite{ghahfarokhiOCELStandardObjectCentric2021}, where each recorded event may refer to multiple objects. 
For example, consider the following OCEL $L_1$ that records activities in a bicycle manufacturing workshop.
For a given order, a frame, a set of wheels, and a handlebar are collected before being assembled:
\begin{eqnarray*}
L_1 = [&\tup{\m{collect},\{f_1, h_1, w_1, w_2\}},
\tup{\m{assemble\_w},\{f_1, w_1, w_2\}},&\\
&
\tup{\m{assemble\_h},\{f_1, h_1\}},\tup{\m{collect},\{f_2, h_2, w_3, w_4\}},&\\
&
\tup{\m{assemble\_h},\{f_2, h_2\}}, \tup{\m{assemble\_w},\{f_2, w_3, w_4\}}
&]
\end{eqnarray*}
with $f_i$ of type \emph{frame}, $h_i$ of type \emph{handlebar}, and $w_i$ of type \emph{wheel}.  Two bicycles are ordered; the first \emph{collect} event takes the frame, handlebar, and two wheels for the first bike; the \texttt{assemble\_w} event assembles two wheels to a frame, etc.

To replay the event log and describe the behavior captured in an OCEL, an \emph{object-centric Petri net (OCPN)} can be discovered~\cite{vanderaalstDiscoveringObjectcentricPetri2020}. For instance,
the OCPN in~\autoref{example:ocpn} allows to replay the behavior of log $L_1$. In the initial marking ($\tikz{\node[starttoken]{}}$), all recorded objects in the log are placed, so that, e.g., all wheel objects $w_1, \ldots w_4$ reside initially in the \emph{wheel} place.
The transitions can fire by binding objects and processing them through the net. In OCPNs, a non-variable arc binds a single object, while a variable arc (double line) binds a set of objects.
Finally, the net accepts a given execution if all objects are placed in the final marking ($\tikz{\node[endtoken]{}}$).
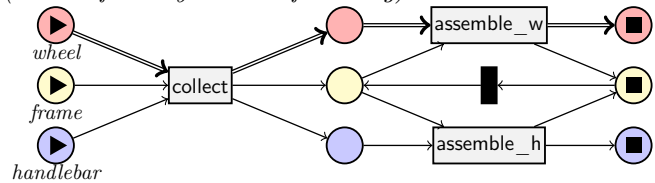
\begin{example}[OCPN for bicycle manufacturing]
\label{example:ocpn}

\centering
\begin{tikzpicture}[node distance=24mm]
\node[wheelPlace] (w0) {} node[below, insc, yshift=-2mm] {$\textit{wheel}$} node[starttoken] {};
\node[wheelPlace, right of=w0, xshift=24mm] (w1) {};
\node[wheelPlace, right of=w1, xshift=24mm] (w2) {} node[endtoken] at (w2) {};
\node[framePlace, below of=w0, yshift=14mm] (f0) {} node[below, insc, yshift=-12mm] {$\textit{frame}$} node[starttoken] at (f0) {};
\node[framePlace, below of=w1, yshift=14mm] (f1) {};
\node[framePlace, below of=w2, yshift=14mm] (f2) {} node[endtoken] at (f2) {};
\node[handlePlace, below of=f0, yshift=14mm] (h0) {} node[below, insc, yshift=-22mm] {$\textit{handlebar}$} node[starttoken] at (h0) {};
\node[handlePlace, below of=f1, yshift=14mm] (h1) {};
\node[handlePlace, below of=f2, yshift=14mm] (h2) {} node[endtoken] at (h2) {};
\node[trans, right of=f0] (cp) {$\m{collect}$};
\node[trans, right of=w1] (aw) {$\m{assemble\_w}$};
\node[trans, right of=h1] (ah) {$\m{assemble\_h}$};
\node[silentTrans, right of=f1] (s1) {};
\draw[fatarc] (w0) -- node[above, insc] {} (cp);
\draw[arc] (h0) -- node[above, insc] {} (cp);
\draw[fatarc] (cp) -- node[above, insc] {} (w1);
\draw[arc] (cp) -- node[above, insc] {} (h1);
\draw[fatarc] (w1) -- node[above, insc] {} (aw);
\draw[fatarc] (aw) -- node[above, insc] {} (w2);
\draw[arc] (f0) -- node[below, insc] {} (cp);
\draw[arc] (cp) -- node[below, insc] {} (f1);
\draw[arc] (f1) -- node[above, insc] {} (aw);
\draw[arc] (h1) -- node[above, insc] {} (ah);
\draw[arc] (ah) -- node[above, insc] {} (h2);
\draw[arc] (f1) -- node[above, insc] {} (ah);
\draw[arc] (aw) -- node[above, insc] {} (f2);
\draw[arc] (ah) -- node[above, insc] {} (f2);
\draw[arc] (f2) -- node[below, insc] {} (s1);
\draw[arc] (s1) -- node[below, insc] {} (f1);
\end{tikzpicture}
\end{example}
OCPNs favor simplicity over expressiveness:
Although discovered OCPNs consider the token flow of multiple object types, OCPNs do not provide mechanisms to track the identity of objects, their relationships, and cardinality constraints.
As pointed out in~\cite{Aal23}, this deliberate underspecification makes discovered OCPNs underfitting, as they allow for too much behavior.
This becomes apparent if the expected data model with cardinality constraints for log $L_1$ is explicitly modeled as a UML class diagram:

{\centering
\resizebox{.9\textwidth}{!}
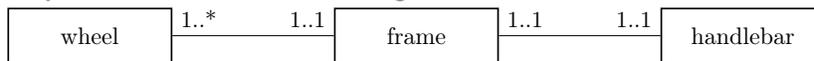
{
\begin{tikzpicture}[node distance=48mm]
\node[class] (wheel) {wheel};
\node[class, right of=wheel] (frame) {frame};
\node[class, right of=frame] (handlebar) {handlebar};
\draw[] (wheel) -- node[above, insc, xshift=-10mm, scale=1.2] {1..*} node[above, insc, xshift=10mm, scale=1.2] {1..1} (frame);
\draw[] (frame) -- node[above, insc, xshift=-10mm, scale=1.2] {1..1} node[above, insc, xshift=10mm, scale=1.2] {1..1} (handlebar);
\end{tikzpicture}}\par}
\noindent
The behavior recorded in $L_1$ satisfies the cardinality constraints for a many-to-one relationship between the wheels and a frame, and a one-to-one relationship between the frame and handlebar.
However, in an OCPN it is impossible to express that a transition should bind objects in accordance with their mutual relationships, and it cannot distinguish transition firings that satisfy this criterion from those that bind unrelated objects, a distinction that is often essential when dealing with synchronization~\cite{gianolaObjectCentricConformanceAlignments2024}. For example, in the OCPN in Example~\ref{example:ocpn} wheels collected for a given frame may be reshuffled, that is, assembled with a different frame.
Hence, it accepts the execution in the following log:
\begin{eqnarray*}
L_2 = & [\tup{\m{collect},\{f_3, h_3, w_5, w_6\}},
\tup{\m{collect},\{f_4, h_4, w_7\}},
\tup{\m{assemble\_w},\{f_3, w_5\}},\\&
\tup{\m{assemble\_h},\{f_3, h_4\}},
\tup{\m{assemble\_w},\{f_4, w_6, w_7\}},
\tup{\m{assemble\_h},\{f_4, h_3\}}
]
\end{eqnarray*}
where wheel $w_6$ is attached to frame $f_4$ even though it was collected with frame $f_3$.
This indicates that whenever an OCPN is used in a prescriptive manner for conformance checking or simulation, further so-called
``bounding behavior'' \cite{Aal23} is needed 
to ensure the desired synchronization.
This bounding behavior can be 
suitably captured through more sophisticated and expressive object-centric process modeling languages, 
such as variants of Petri nets with identifiers (PNIDs) \cite{ghilardiPetriNetbasedObjectcentric2022a,vanderwerfCorrectnessNotionsPetri2024}, or synchronous proclets \cite{fahlandDescribingBehaviorProcesses2019}.
Recently, object-centric Petri nets with identifiers (OPIDs \cite{gianolaObjectCentricConformanceAlignments2024}) have been proposed as a formalism that unifies the features of OCPNs, PNIDs, and synchronous proclets.
In contrast to OCPNs, OPIDs are appropriate for prescriptive purposes because of their high expressiveness.
However, while conformance for OPIDs is handled in \cite{gianolaObjectCentricConformanceAlignments2024}, the discovery of OPIDs and related formalisms has not yet been studied.
To avoid reshuffling like in the above example, and similar synchronization problems,
we start by identifying an important assumption that is implicit in OCELs:
executions are supposed to satisfy
\emph{stable many-to-one relationship} between two object types
which synchronize the corresponding objects throughout the execution. This means that if the relationship between object types $\sigma_m$ and $\sigma_o$ is many-to-one, then whenever an object of type $\sigma_m$ is bound to an object of type $\sigma_o$, 
it is never bound to a different object of type $\sigma_o$.
This implicit assumption comes from the fact that OCEL assumes global and rigid relationships throughout execution~\cite{knoppDiscoveringObjectCentricProcess2023,bertiOCELObjectCentricEvent2024}.

To find a model that satisfies both the simplicity of discovered OCPNs and the expressiveness to enforce the desired bounding behavior, this paper presents a rigorous mapping from an OCPN together with a set of stable many-to-one relationships, to an equivalent \emph{replay OPID} that enforces the intended synchronization.
By doing so, we at once provide a fourfold contribution:
\begin{inparaenum}[(1)]
\item we provide a first discovery technique for one specific class of OPIDs,
\item we identify the class of relatively simple replay OPIDs that are yet expressive enough to enforce the desired bounding behaviour, 
\item we formally prove that a so-obtained replay OPID is equivalent to the original OCPN for executions that respect the data model
\item 
we provide a proof-of-concept implementation of our transformation.
\end{inparaenum}

To this end, after recapping the preliminaries (\autoref{sec:preliminaries}), we first define a mapping of an OCPN to an equivalent replay OPID
that tracks object identifiers but not object relationships 
(\autoref{sec:equiv}).
In a second step, given a set of stable many-to-one relationships, we define a translation that encodes the synchronization semantics into a more restrictive replay OPID (\autoref{sec:sync}), which we prove to be equivalent to the original OCPN in replaying OCELs that satisfy the stable many-to-one relationships, while prohibiting violating executions.
To illustrate practical relevance, we discuss our contribution in the context of related work, show that stable many-to-one relationships frequently occur in existing OCEL benchmarks, and describe our implementation of the transformations (\autoref{sec:related_work}).

\section{Preliminary Definitions}\label{sec:preliminaries}

\subsection{Object-centric Event Logs and Object-centric Petri Nets}
To define OCELs and OCPNs, we recall notions by van der Aalst and Berti~\cite{vanderaalstDiscoveringObjectcentricPetri2020}.

\begin{definition}
Given a set of object types $\objtypes$, a set of activities $\activities$, and a set of timestamps $\mathbb{T}$ under a total order $<$, an \emph{object-centric event log} (OCEL) is a tuple $L = ( E, \objects, \pi_{act}, \pi_{obj}, \pi_{time} )$ where:
\begin{itemize}
    \item $E$ is a set of events;
    \item $\objects$ is a set of object identifiers, each typed by the function $ot: \objects \to \objtypes$;
    \item the functions $\pi_{act}: E \to \activities$, $\pi_{obj}: E \to \mathcal{P}(\objects)$, and $\pi_{time}: E \to \mathbb{T}$, mapping each event to an activity, a set of objects, and a timestamp.
\end{itemize}
\end{definition}

For example, the event logs $L_1$ and $L_2$ in \autoref{sec:intro} are defined for the object types \textit{Wheel}, \textit{Frame}, and \textit{Handlebar}. The event $e = \tup{\m{collect},\{f_1, h_1, w_1, w_2\}}$ in $L_1$ refers to the activity $\pi_{act}(e) = \texttt{collect}$ and the objects $\pi_{obj}(e) = \{f_1, h_1, w_1, w_2\}$, where $ot(f_1) = \textit{Frame}$.

\begin{definition}
\label{def:OCPN}
Given a set of object types $\Sigma$, and a set of activities $\activities$, an \emph{object-centric Petri net} (OCPN) is defined as a tuple  
$N = (P, T, F, pt, F_{var}, \ell)$,
where:
\begin{enumerate}
    \item $P$ and $T$ are finite sets of places and transitions such that $P\cap T=\emptyset$,
    \item $F \subseteq (P \times T) \cup (T \times P)$ is the flow relation between places and transitions,
    \item $pt\colon P \rightarrow \Sigma$ maps every place to an object type,
    \item $F_{var} \subseteq F$ is the set of variable flow relations,
    \item $\ell \colon T \to \activities \cup \{\tau\}$ is the transition labelling where $\tau$ marks an invisible activity.
\end{enumerate}
\end{definition}

The pre-places of a transition $t$ are all places with ingoing flow relations $\pre t = \{p \in P \mid (p,t) \in F\}$, while the post-places of $t$ are all those connected with outgoing flow relations $\post{t} = \{p \in P \mid (t,p) \in F\}$.
For a transition $t$, the function $tpl(t)$ returns the types of all its corresponding places $p \in \pre t \cup t \post{}$. $tpl_{var}(t)$ returns all place types connected via variable flow, and $tpl_{nv}(t)$ all place types connected via non-variable flow.
An OCPN is considered well-formed if all flow relations of one transition agree in the variability for the place types, i.e., $\forall t \in T: tpl_{var}(t) \cap tpl_{nv}(t) = \emptyset$.

All possible tokens and their location in a place are defined as $Q = \{(p,o) \in P \times \objects \mid pt(p) = type(o)\}$. A marking is generally defined as a multiset of token-object pairs $M \colon Q \to \mathbb{N}$~\cite{vanderaalstDiscoveringObjectcentricPetri2020}.
However, we pose the following specific assumptions about the structure of an OCPN.
The contemporary discovery algorithms for OCPN~\cite{vanderaalstDiscoveringObjectcentricPetri2020,vandettenDiscoveringCompactLive2024} utilize the Inductive Miner
algorithm~\cite{leemansScalableProcessDiscovery2018}, which results in sound workflow nets for each object type.
As a result, all places are bound to at most one object identity by design.
Therefore, we interpret the marking $M$ of an OCPN as a set of objects in a marking such that $M\colon Q \to \{0,1\}$.
Furthermore, each object type has one initial place and one final place.

The set of binding executions is $B = \{ (t, b) \in T \times (\objtypes \to \mathcal{P}(\objects)) \mid dom(b) = tpl(t) \wedge \forall \sigma \in tpl_{nv}(t): |b(\sigma)| = 1 \}$. It contains pairs of transition $t$ and binding $b$. A binding $b$ is a function that maps an object type $\sigma \in \objtypes$ to a set of objects.
We denote by $\codom(b)$ the codomain of $b$, that is, $\codom(b) = \bigcup_{\sigma \in \objtypes} b(\sigma)$.
This binding execution consumes a set of object tokens from the transition's pre-places $cons(t,b) = \{(p,o) \in Q \mid p \in \pre t \wedge o \in b(pt(p))\}$. Similarly, it produces object tokens $prod(t,b) = \{(p,o) \in Q \mid p \in {\post t} \wedge o \in b(pt(p))\}$.
This binding is enabled in a marking $M_i$ and can be executed if $cons(t,b) \leq M_i$.
The resulting marking is $M_j = M_i - cons(t,b) + prod(t,b)$. The firing is denoted as $\smash{M_i \goto{(t,b)} M_j}$.
Note that by the definition of bindings, if a transition $t$ has two places in $p_1,p_2 \in \pre t$ of the same object type, then $t$ can only be fired by consuming the same objects from $p_1$ and $p_2$. 
A binding $(t, b)$ for a transition $t$ considers only objects $b \subseteq \objects$ that belong to the types of the respective input places and are in line with the variable and non-variable arcs.

An \emph{accepting OCPN} combines an OCPN $N$ with an initial marking $\Minit \in \mathcal{B}(Q)$ and a final marking $\Mfinal \in \mathcal{B}(Q)$.
An accepted run is a sequence of binding executions from initial to final marking $\Minit \goto{*} \Mfinal$.
The accepting OCPN can be used to \emph{replay} the behavior of the log. All recorded objects are interpreted as the initial marking $\Minit$. An event in the log and its recorded objects are interpreted as a binding execution of the corresponding transition and the involved objects, which results in a new marking. After replaying all events consecutively, the OCPN can accept the log if the run results in the final marking $\Mfinal$. If an event has no corresponding binding execution or $\Mfinal$ is not reached, the log is rejected by the OCPN.

\subsection{Object-centric Petri Nets with Identifiers}
\begin{example}[Simple OPID for bike manufacturing] The OPID combines the object types \textit{Wheel} (red), \textit{Frame} (yellow), and relations for both (orange places).
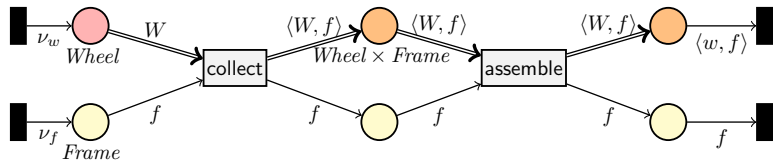
\label{example:opid}

\centering
\begin{tikzpicture}[node distance=24mm]
\node[silentTrans] (ew) {};
\node[wheelPlace, right of=ew, xshift=-12mm] (w0) {} node[below of=w0, yshift=20mm, insc] {$\textit{Wheel}$};
\node[wheelFramePlace, right of=w0, xshift=24mm] (w1) {}node[below of=w1, yshift=20mm, insc] {$\textit{Wheel}\times \textit{Frame}$};
\node[wheelFramePlace, right of=w1, xshift=24mm] (w2) {};
\node[silentTrans, right of=w2, xshift=-8mm] (cw) {};
\node[silentTrans, below of=ew, yshift=8mm] (ef) {};
\node[framePlace, below of=w0, yshift=8mm] (f0) {} node[below of=f0, yshift=20mm, insc] {$\textit{Frame}$};
\node[framePlace, below of=w1, yshift=8mm] (f1) {};
\node[framePlace, below of=w2, yshift=8mm] (f2) {};
\node[silentTrans, below of=cw, yshift=8mm] (cf) {};
\node[trans, below of=w0,yshift=17mm,xshift=24mm] (rec) {$\m{collect}$};
\node[trans, below of=w1,yshift=17mm,xshift=24mm] (aw) {$\m{assemble}$};
\draw[arc] (ew) -- node[below, insc] {$\nu_w$} (w0);
\draw[arc] (ef) -- node[below, insc] {$\nu_f$} (f0);
\draw[arc] (w2) -- node[below, insc] {$\tup{w,f}$} (cw);
\draw[arc] (f2) -- node[below, insc] {$f$} (cf);
\draw[fatarc] (w0) -- node[above, insc] {$W$} (rec);
\draw[fatarc] (rec) -- node[above, insc] {$\tup{W,f}$} (w1);
\draw[fatarc] (w1) -- node[above, insc] {$\tup{W,f}$} (aw);
\draw[fatarc] (aw) -- node[above, insc] {$\tup{W,f}$} (w2);
\draw[arc] (f0) -- node[below, insc] {$f$} (rec);
\draw[arc] (rec) -- node[below, insc] {$f$} (f1);
\draw[arc] (f1) -- node[below, insc] {$f$} (aw);
\draw[arc] (aw) -- node[below, insc] {$f$} (f2);
\end{tikzpicture}
\end{example}

We summarize the definitions from~\cite{gianolaObjectCentricConformanceAlignments2024}.
Every object type $\sigma \in \objtypes$ is assumed to have a domain $\dom(\sigma) \subseteq \objects$,
given by all objects in $\objects$ of type $\sigma$. 
List types with a base type in $\sigma$ are denoted as $\listtype{\sigma}$.
Each Petri net place has a \emph{color} that is a cartesian product of data types from $\objtypes$; precisely, the set of colors $\colors$ is the set of all $\sigma_1 \times \cdots \times \sigma_m$ such that $m \geq 1$ and $\sigma_i \in \objtypes$ for all $1\leq i\leq m$.
We fix a set of $\objtypes$-typed variables $\allvars = \varset \uplus \varset_{list} \uplus \nuvarset$ as the disjoint union of 
a set $\varset$ of ``normal'' variables that refer to single objects, denoted by lower-case letters like $x$, with a type $\vartype(x) \in \Sigma$,
a set $\varset_{list}$ of list variables that refer to a list of objects of the same type, denoted by upper case letters like $U$, with a type $\vartype(U)=\listtype{\sigma}$ for some $\sigma \in \Sigma$,
and 
 a set $\nuvarset$ of variables referring to fresh objects, denoted as $\nu$, 
 with $\vartype(\nu) \in \Sigma$. 

In OPIDs~\cite{gianolaObjectCentricConformanceAlignments2024}, tokens are object tuples associated with a color.
To define relationships between objects in consumed and produced tokens when firing a transition, arc \emph{inscriptions} are defined.

\begin{definition}
\label{def:inscription}
An \emph{inscription} is a tuple $\vec x = \tup{x_1, \dots, x_m}$
such that $m \geq 1$ and $x_i\in \allvars$ for all $i$, but at most one $x_i \in \listvarset$, for $1 \leq i \leq m$.
We call $\vec x$ a 
\emph{template inscription} if $x_i\in \listvarset$ for some $i$, and a
\emph{simple inscription} otherwise.
\end{definition}

Template inscriptions capture an arbitrary number of tokens of the same color; they are used to label the correspondent of variable arcs in OCPNs.
The color of an inscription $\iota =\tup{x_1, \dots, x_m}$ is given by $\coloring(\iota) =\tup{\sigma_1, \dots, \sigma_m}$ where $\sigma_i=\vartype(x_i)$ if $x_i\in \varset\cup \nuvarset$, and $\sigma_i=\sigma'$ if $x_i$ is a list variable of type $\listtype{\sigma'}$.
We set  $\vars{\iota} = \{x_1, \dots, x_m\}$.
The set of all inscriptions is denoted $\Omega$.
Based on this, an OPID is formally defined as follows~\cite{gianolaObjectCentricConformanceAlignments2024}.

\begin{definition}
\label{def:OCInet}
An \emph{object-centric Petri net with identifiers} (OPID) 
is defined as a tuple  
$N = (\objtypes, P, T, \inflow, \outflow, \coloring,\ell)$,
where:
\begin{enumerate}
\item $P$ and $T$ are finite sets of places and transitions such that $P\cap T=\emptyset$,
\item $\coloring\colon P \rightarrow \colors$ maps every place to a color,
\item $\ell \colon T \to \activities \cup \{\tau\}$ is the transition labelling where $\tau$ marks an invisible activity,
\item $\inflow \colon P \times T \rightarrow \Omega$ is a partial function called \emph{input flow},  such that\\
$\coloring(\inflow(p,t))=\coloring(p)$
for every $(p,t)\in \dom(\inflow)$,
\item $\outflow\colon T \times P \rightarrow \Omega$ is a partial function called \emph{output flow}, such that\\
$\coloring(\outflow(t,p))=\coloring(p)$ for every $(t,p)\in \dom(\outflow)$.
\end{enumerate}
We set $\invars{t} = \bigcup_{p \in P} \vars{\inflow(p,t)}$, and
$\outvars{t} = \bigcup_{p \in P} \vars{\outflow(t,p)}$, and require that
$\invars{t}\,{\cap}\,\nuvarset\,{=}\,\emptyset$ and 
$\outvars{t} \subseteq \invars{t}\,{\cup}\,\nuvarset$, for all $t\,{\in}\,T$.
\end{definition}

If $\inflow(p, t)$ is defined, it
is called a \emph{variable} flow if $\inflow(p,t)$ is a template inscription,
and \emph{non-variable} flow otherwise; and similar for output flows.
Variable flows play the role of variable arcs in OCPNs~\cite{vanderaalstDiscoveringObjectcentricPetri2020}; they can carry multiple tokens at once.
The common notations $\pre t = \{p \mid (p,t)\in \dom(\inflow)\}$ and $\post t = \{p \mid (t,p)\in \dom(\outflow)\}$ are used as well.

\textbf{Semantics.}
Given the set of objects $\objects$,
the set of \emph{tokens} $\tokens$ is the set of object tuples $\tokens\,{=}\,\{\objects^m\,{\mid}\,m\,{\geq}\,1\}$.
The \emph{color} of a token $\omega\in \tokens$ of the form $\omega=\tup{o_1, \dots, o_m}$ is denoted
$\coloring(\omega)=\tup{\vartype(o_1), \dots, \vartype(o_m)}$.
A \emph{marking} of an OPID $N=\tup{\objtypes, P, T, \inflow, \outflow, \coloring,\ell}$ is a function $M\colon P\rightarrow 2^{\tokens}$, such that for all $p\in P$ and $\tup{o_1,\dots o_m} \in M(p)$, it holds that $\coloring(\tup{o_1,\dots o_m}) = \coloring(p)$. 
Let $\mathit{Lists}(\objects)$ denote the set of objects lists of the form $[o_1, \dots, o_k]$ with $o_1,\dots, o_k \in \objects$ such that all $o_i$ have the same type;
the type of such a list is then $\listtype{\vartype(o_1)}$.
Next, \emph{bindings} are defined to fix which objects are involved in a transition firing.

\begin{definition}
A \emph{binding} for a transition $t$ and a marking $M$ is a type-preserving function 
$\beta\colon \invars{t} \cup \outvars{t} \to \objects \cup \mathit{Lists}(\objects)$.
To ensure the freshness of created values, $\beta$ is required to be injective on $\nuvarset \cap \outvars{t}$, and $\beta(\nu)$ must not occur in $M$ for all $\nu \in \nuvarset \cap \outvars{t}$.
\end{definition}

We denote by $\codom(\beta) \subseteq \objects$ the set of all objects occurring in tuples returned by $\beta$.
Bindings are extended to inscriptions, denoted $\vec \beta$, to fix which tokens participate in a transition firing.
For an inscription $\iota\,{=}\,\tup{x_1, \dots, x_m}$, let $o_i\,{=}\,\beta(x_i)$ for all $1\,{\leq}\,i\,{\leq}\,m$.
Then $\vec \beta(\iota)$ is the set of object tuples defined as follows:
if $\iota$ is a simple inscription, then
$\vec \beta(\iota) = \{\tup{o_1, \dots, o_m}\}$.
Otherwise, there must be one $x_i$, $1\,{\leq}\,i\,{\leq}\,n$, such that $x_i \in \listvarset$, and consequently $o_i$ must be a list, say $o_i=[u_1, \dots, u_k]$ for some $u_1, \dots, u_k$.
Then $\vec \beta(\iota) = \{\tup{o_1, \dots, o_{i-1}, u_1, o_{i+1},\dots, o_m}, \dots, \tup{o_1, \dots, o_{i-1}, u_k, o_{i+1},\dots, o_m}\}$.
$\mathcal B$ denotes the set of all bindings.
Finally, a transition with a binding $\beta$ is enabled if all object tuples pointed by $\vec \beta$ occur in the current marking:

\begin{definition}
\label{def:enabled}
A transition $t \in T$ and a binding $\beta$ for marking $M$ are \emph{enabled} in $M$
if $\vec \beta(\inflow(p,t)) \subseteq M(p)$ for all $p \in \pre{t}$.
\end{definition}


\begin{definition}
Let transition $t$ be enabled in marking $M$ with binding $\beta$.
The \emph{firing} of $t$ yields the new marking $M'$ given by $M'(p)=M(p) \setminus \vec \beta(\inflow(p,t)) $ for all $p \in \pre t\setminus \post t$, and
$M'(p)=M(p) \cup \vec \beta(\outflow(p,t))$ for all $p \in {\post t}\setminus {\pre t}$, and $M'(p)=M(p) \setminus \vec \beta(\inflow(p,t)) \cup \vec \beta(\outflow(p,t))$ for all $p \in \pre t\cap \post t$.%
\footnote{The case for $p \in \pre t\cap \post t$ was omitted in~\cite{gianolaObjectCentricConformanceAlignments2024}, which is insufficient for cyclic nets.}
\end{definition}

We write $M \goto{t,\beta} M'$ to denote that $t$ is enabled for binding $\beta$ in $M$, and its firing yields $M'$.
An \emph{accepting} OPID $N$ is an OPID 
together with a set of initial markings $M_{\mathit{init}}$ and a set of final markings $M_{\mathit{final}}$.
A run is accepted if it starts in the initial and ends in the final marking $\Minit \goto{*} \Mfinal$.

\section{Equivalent Object-centric Petri Nets with Identifiers}
\label{sec:equiv}

This section defines the mapping $\directmap$ from an OCPN $N$ to an OPID. We show that the replays of both nets $N$ and $\directmap(N)$ on a given event log are equivalent, that is, they accept and reject the same executions.

In this paper, we assume that for any given OCPN
$N = (P,T,F,pt,F_{var},\ell)$
there are fixed sets of places $\Pplay, \Pstop \subseteq P$ called \emph{play} and \emph{stop places}, 
respectively, such that for each $\sigma\in\Sigma$ there is exactly one place $p\in \Pplay$ and $p' \in \Pstop$ such that $pt(p)=pt(p')=\sigma$.
These are the places that are supposed to be populated by initial and final markings, respectively.
Markings $\Minit$ and $\Mfinal$ are \emph{suitable} initial and final markings for an OCPN with play and stop places $\Pplay, \Pstop$
if they put tokens only in those places, i.e., for all $o\in\objects$, $\Minit(p,o) = 0$ for all $p\not\in\Pplay$ and $\Minit(p',o) = 0$ for all $p'\not\in\Pstop$, and similar for $\Mfinal$.

Note that a given set of objects uniquely determines suitable initial and final markings for $N$:
precisely,
for a finite set of objects $O \subseteq \objects$, let $\Minit^O$ (resp. $\Mfinal^O$) be the marking such that
for each $\sigma\in\Sigma$ and and each $o\in O_\sigma$, $\Minit^O(p_\sigma,o)=1$ (resp. $\Mfinal^O(p_\sigma,o)=1$) where $p_\sigma$ is the unique place in $\Pplay$ (resp. $\Pstop$) of type $\sigma$.
Our mapping $\directmap$ maps an OCPN $N$ with start and stop places to an OPID $\directmap(N)$
that keeps all the places and transitions of $N$. The flow relations are adjusted and extended with arc inscriptions. Moreover, $\directmap(N)$ has additional transitions for creating and removing objects.
In this way, the OPID can, on demand, populate the places in $\Pplay$ with objects in an initial marking for $N$ and consume objects from the places in $\Pstop$ once a final marking for $N$ is reached.
The initial and final markings for $\directmap(N)$ are always empty.

Precisely, the OPID has an emitting transition for every object type, i.e., we set $T_e := \{ t_{e,\sigma} \mid \sigma \in \objtypes \}$; and a consuming transition for every object type,
i.e., we set $T_c := \{ t_{c,\sigma} \mid \sigma \in \objtypes \}$.
With its emitting transitions $T_e$, an OPID allows the creation of arbitrary new objects. However, in the replay of an OCEL $L$, this set of objects is fixed and denoted as $\objects_L$. 
Hence, when replaying the log $L$, 
the inscription variables $\nu_i \in \outflow(T_e)$ will be bound exclusively to objects in $\objects_L$.
To be able to define suitable inscriptions, we assume that for every object type $\sigma \in \Sigma$ there is a normal variable $x_\sigma \in \varset$, a list variable $X_\sigma \in \listvarset$, and a fresh object variable $\nu_\sigma \in \nuvarset$. We then set:
\begin{xalignat*}{2}
\outflow&\colon (T \cup T_e) \times P \to \Omega &
\outflow&(t,p) = 
    \begin{cases}
\tup{x_{pt(p)}}
    & \text{if } (t,p) \in F \setminus F_{var} \\
\tup{X_{pt(p)}}
    & \text{if } (t,p) \in F_{var} \\
\tup{\nu_{pt(p)}}
    & \text{if } p \in \Pplay \wedge t \in T_e 
  \end{cases}
\\
\inflow&\colon P \times (T \cup T_c) \to \Omega &
\inflow&(p,t) = 
  \begin{cases}
\tup{x_{pt(p)}}
    & \text{if } (p,t) \in F \setminus F_{var} \\
\tup{X_{pt(p)}}
    & \text{if } (p,t) \in F_{var} \\
    \tup{x_{pt(p)}}, & \text{if } p \in \Pstop \wedge t \in T_c
  \end{cases}
\end{xalignat*}
The equivalent OPID is given by $\directmap(N) = (\Sigma_1, P_1, T_1, F_{in}, F_{out}, \textit{color}, \ell_1)$ where $\Sigma_1=\Sigma$, $P_1=P$, $T_1=T \cup T_e \cup T_c$, $\textit{color}(p) = \tup{pt(p)}$ for all $p\in P$, and

\begin{xalignat*}{2}
\ell_1 & \colon T \to \activities &
\ell_1&(t) = 
  \begin{cases}
    \tau, & \text{if } t \in T_e \cup T_c \\
    \ell(t), & \text{otherwise}
  \end{cases}
\end{xalignat*}

\begin{example}[Equivalent OPID for bike manufacturing]
The following OPID is obtained from the OCPN in ~\autoref{example:ocpn} by applying the translation $\directmap$.

\centering
\begin{tikzpicture}[node distance=24mm]
\node[silentTrans] (ew) {};
\node[wheelPlace, right of=ew, xshift=-12mm] (w0) {} node[pseudostart] at (w0) {};
\node[wheelPlace, right of=w0, xshift=24mm] (w1) {};
\node[wheelPlace, right of=w1, xshift=24mm] (w2) {} node[pseudoend] at (w2) {};
\node[silentTrans, right of=w2, xshift=-12mm] (cw) {};
\node[silentTrans, below of=ew, yshift=14mm] (ef) {};
\node[framePlace, below of=w0, yshift=14mm] (f0) {} node[pseudostart] at (f0) {};
\node[framePlace, below of=w1, yshift=14mm] (f1) {};
\node[framePlace, below of=w2, yshift=14mm] (f2) {} node[pseudoend] at (f2) {};
\node[silentTrans, below of=cw, yshift=14mm] (cf) {};
\node[silentTrans, below of=ef, yshift=14mm] (eh) {};
\node[handlePlace, below of=f0, yshift=14mm] (h0) {} node[pseudostart] at (h0) {};
\node[handlePlace, below of=f1, yshift=14mm] (h1) {};
\node[handlePlace, below of=f2, yshift=14mm] (h2) {} {} node[pseudoend] at (h2) {};
\node[silentTrans, below of=cf, yshift=14mm] (ch) {};
\node[trans, right of=f0] (cp) {$\m{collect}$};
\node[trans, right of=w1] (aw) {$\m{assemble\_w}$};
\node[silentTrans, right of=f1] (s1) {};
\node[trans, right of=h1] (ah) {$\m{assemble\_h}$};
\draw[arc] (ew) -- node[below, insc] {$\nu_w$} (w0);
\draw[arc] (ef) -- node[below, insc] {$\nu_f$} (f0);
\draw[arc] (eh) -- node[below, insc] {$\nu_h$} (h0);
\draw[arc] (w2) -- node[below, insc] {$w$} (cw);
\draw[arc] (f2) -- node[below, insc] {$f$} (cf);
\draw[arc] (h2) -- node[below, insc] {$h$} (ch);
\draw[fatarc] (w0) -- node[above, insc] {$W$} (cp);
\draw[fatarc] (cp) -- node[above, insc] {$W$} (w1);
\draw[fatarc] (w1) -- node[above, insc] {$W$} (aw);
\draw[fatarc] (aw) -- node[above, insc] {$W$} (w2);
\draw[arc] (f0) -- node[below, insc] {$f$} (cp);
\draw[arc] (cp) -- node[below, insc] {$f$} (f1);
\draw[arc] (f1) -- node[above, insc] {$f$} (aw);
\draw[arc] (aw) -- node[above, insc] {$f$} (f2);
\draw[arc] (f1) -- node[below, insc] {$f$} (ah);
\draw[arc] (ah) -- node[below, insc] {$f$} (f2);
\draw[arc] (f2) -- node[below, insc] {$f$} (s1);
\draw[arc] (s1) -- node[below, insc] {$f$} (f1);
\draw[arc] (h0) -- node[below, insc] {$h$} (cp);
\draw[arc] (cp) -- node[below, insc] {$h$} (h1);
\draw[arc] (h1) -- node[above, insc] {$h$} (ah);
\draw[arc] (ah) -- node[above, insc] {$h$} (h2);
\end{tikzpicture}\par
\end{example}

$\directmap(N)$ is interpreted as an accepting OPID where the initial and final markings are defined as the empty marking $\Mempty$, which requires that all places be empty.
The example OPID can replay the logs $L_1$ as well as $L_2$ as it has equivalent semantics to \autoref{example:ocpn} in combining objects without additional constraints.

As formalized in \autoref{theorem:equal:1}, the execution semantics of the OCPN $N$ and the mapped OPID $\directmap(N)$ are equivalent for replaying any OCEL. This theorem can be proven by considering the equivalent definition of binding executions. Considering $M_i'=\directmap(M_i)$ as the mapped markings and $\beta_i = \directmap(b_i)$ as the mapped bindings, $\directmap(N)$ can produce and consume the same set of objects $\objects$ as specified in the initial and final marking of the OCPN.
For reasons of space, the full proof of \autoref{theorem:equal:1} is moved to Appendix~\ref{app:a}. 
\begin{theorem}
\label{theorem:equal:1}
Let $\objects$ be a set of objects, $N$ be an accepting OCPN, $N' = \directmap(N)$, and $\Mempty$ denote the empty marking of $N'$.
There exists a run of $N$:
\begin{equation*}
M_0 \goto{(t_1,b_1)} M_1 \goto{(t_2,b_2)} \dots \goto{(t_n,b_n)} M_n
\end{equation*}
with $M_0{=}\Minit^\objects$ and $M_n{=}\Mfinal^\objects$
if and only if there is a run of $N'$ over the set of objects $\objects$:
\begin{equation*}
\Mempty \goto{*} M_0' \goto{(t_1,\beta_1)} M_1' \goto{(t_2,\beta_2)} \dots \goto{(t_n,\beta_n)} M_n' \goto{*} \Mempty
\end{equation*}
such that $M_i'=\directmap(M_i)$ for all $0 \leq i \leq n$ and $\beta_i = \directmap(b_i)$ for all $1\leq i \leq n$,
such that $\Mempty \goto{*} M_0'$ is an emitting sequence, $M_n' \goto{*} \Mempty$ a collapsing sequence, and no emitting or consuming transitions occur in between.
\end{theorem}

\section{Synchronization for Stable Many-to-One Relationships}
\label{sec:sync}

In this section, we define what constitutes a stable many-to-one relationship for object-centric process executions.
Given an OCEL that assumes the stability of relationships, such a stable many-to-one relationship implies that objects on the many-side will always be bound to the same object of the one-side. Hence, the processing of such objects is implicitly synchronized.
We then extend the mapping $\directmap$ from \autoref{sec:equiv} to implement the expected synchronization for a set of defined stable many-to-one relationships.
Finally, we show that any execution replay of the OCPN is accepted by the mapped OPID if and only if it conforms to the stable many-to-one relationships.

In the following, let a pair $(\sigma_m,\sigma_o)$ be an \emph{object type relationship} if $\sigma_m \neq \sigma_o$.
For example, for the bike manufacturing example, a stable many-to-one relationship may exist for \textit{(Wheel, Frame)}.
Given an OCEL $L=(E,\objects,\pi_{act},\pi_{obj},\pi_{time})$, we define the function $lo \colon \objects \times \objtypes \to \objects$, which, given an object and an object type, returns all co-occurring objects of this type, i.e. $lo(o, \sigma) = \bigcup_{e \in E} \{ o' \mid o, o' \in \pi_{obj}(e) \wedge ot(o') = \sigma\}$.
For OCPNs and OPIDs, $lo$ is defined accordingly over a run of the net and bindings.

\begin{definition}
\label{def:conform}
An object type relationship $(\sigma_m, \sigma_o)$ is a \emph{stable many-to-one relationship} in an event log $L$ if for all objects $o_m, o_o \in \objects$, where $ot(o_m) = \sigma_m$ and $ot(o_o) = \sigma_o$ it holds that $|lo(o_m, \sigma_o)| = 1$.
\end{definition}

Similarly for OCPNs and OPIDs, $(\sigma_m, \sigma_o)$ is a \emph{stable many-to-one relationship} if the same holds for a run in the net:
The objects $o_m$ of the many-side $\sigma_m$ in the stable relationship $(\sigma_m, \sigma_o)$ are required to be linked to exactly one object $o_o$ of the one-side $\sigma_o$. On the other hand, every $o_o$ must have a non-empty set of linked objects $o_m$.
For example, \textit{(Wheel, Frame)} is a stable many-to-one relationship in the log $L_1$ from \autoref{sec:intro}, but not in $L_2$.
Note that the one-to-one relationship between handlebar and frame manifests itself as two stable many-to-one relationships \textit{(Handlebar, Frame)} and \textit{(Handlebar, Frame)}.

\subsection{Mapping to Synchronizing OPIDs}
\label{sec:mapping-rigid}
The goal of this mapping is to enforce a set of stable many-to-one relationships with synchronization behavior.
To this end, for every many-to-one relationship $(\sigma_m,\sigma_o)$, the required set of links $(m,o) \in \objects^2$ of type $(\sigma_m,\sigma_o)$ must be determined \emph{before} the actual execution starts, and must be stable during the whole execution.
In general, the idea is to store all links of a relationship in a dedicated place that is accessed by all transitions that operate on both object types.
For the Petri net used for the manufacturing of bikes, the translation results in the OPID shown in \autoref{example:opid:sync}. For the log $L_1$, it determines the links between wheels $w_1$, $w_2$ of type $W$ and frame $f_1$ of type $F$ as $(w_1,f_1),(w_2,f_1)$ before execution starts and they are not changed.
The transition's synchronization semantics ensure that only objects in a stable relationship are within one binding.

We assume a given OCPN $N = (P,T,F,pt,F_{var},\ell)$ and a fixed set of object type relationships $R$, which are supposed to work as stable many-to-one relationships. 
The mapping $\linkmap$ is defined on top of the mapping $\directmap$ in~\autoref{sec:equiv}. Let $N_1 = \directmap(N) = (\Sigma_1, P_1, T_1, F_{in,1}, F_{out,1}, \coloring_1, \ell_1)$.
In $\linkmap(N_1)$, we add the following new places and transitions:
\begin{compactitem}
\item a set of \emph{link places} $P_L := \{ p_{(\sigma_m,\sigma_o)} \mid (\sigma_m,\sigma_o) \in R \}$, storing the links for each relationship $(\sigma_m,\sigma_o) \in R$, such that $\coloring_2(p_{(\sigma_m,\sigma_o)}) = \tup{\sigma_m,\sigma_o}$,
\item a set of \emph{link transitions} $T_L := \{ \tlink{\sigma_m, \sigma_o} \mid (\sigma_m, \sigma_o) \in R \}$ that creates links for a set of objects of type $\sigma_m$ and one object of $\sigma_o$ to populate $P_L$,
\item new emitting transitions $T_i := \{ t_{i,\sigma} \mid \exists \sigma'\,{\in}\,\objtypes: (\sigma, \sigma')\,{\in}\,R \vee (\sigma', \sigma)\,{\in}\,R\}$,
\item a set $P_b := \{ p_{\sigma,(\sigma_m,\sigma_o)} \mid (\sigma = \sigma_m \vee \sigma = \sigma_o) \wedge (\sigma_m, \sigma_o) \in R \}$  of auxiliary places before the link creation such that $\coloring_2(p_{\sigma,(\sigma_m,\sigma_o)}) = \tup{\sigma}$, and 
\item a set
$P_a := \{ p_{(\sigma_m,\sigma_o),\sigma} \mid (\sigma = \sigma_m \vee \sigma = \sigma_o) \wedge (\sigma_m, \sigma_o) \in R \}$ of auxiliary places after link creation such that $\coloring_2(p_{\sigma,(\sigma_m,\sigma_o)}) = \tup{\sigma}$.
\end{compactitem}
The arrangement of the additional places and transitions is schematically shown below for a pair of object types $(\sigma_m,\sigma_o) \in R$. Here, the blue part of the net represents $N_1$, while the black parts with red labels are extensions by $\linkmap$.

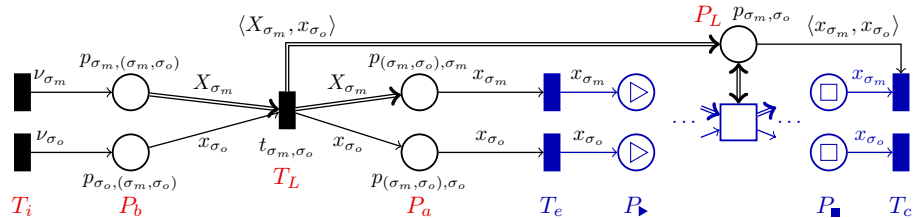
{\centering
\begin{tikzpicture}[node distance=22mm]
\tikzstyle{original}=[blue!70!black]
\node[silentTrans] (im) {};
\node[silentTrans, below of=im, yshift=12mm] (io) {};
\node[place, right of=im,xshift=-4mm] (mb) {};
 \node[yshift=4mm, insc] at (mb) {$p_{\sigma_m,(\sigma_m,\sigma_o)}$};
\node[place, right of=io,xshift=-4mm] (ob) {};
\node[yshift=-4mm, insc] at (ob) {$p_{\sigma_o,(\sigma_m,\sigma_o)}$};
\node[silentTrans, right of=ob, yshift=7mm,xshift=4mm] (cmo) {};
\node[yshift=-5mm, insc] at (cmo) {$t_{\sigma_m,\sigma_o}$};
\node[place, right of=mb, xshift=26mm] (ma) {};
\node[yshift=4mm, insc] at (ma) {$p_{(\sigma_m,\sigma_o), \sigma_m}$};
\node[place, right of=ob, xshift=26mm] (oa) {};
\node[yshift=-4mm, insc] at (oa) {$p_{(\sigma_m,\sigma_o), \sigma_o}$};
\node[silentTrans, right of=ma,original] (em) {};
\node[place, right of=em,xshift=-8mm,original] (m0) {} node[pseudostart,original] at (m0) {};
\node[silentTrans, below of=em, yshift=12mm,original] (eo) {};
\node[place, below of=m0, yshift=12mm,original] (o0) {} node[pseudostart,original] at (o0) {};
\node[place, right of=m0,xshift=10mm,original] (mf) {} node[pseudoend,original] at (mf) {};
\node[place, below of=mf, yshift=12mm,original] (of) {} node[pseudoend,original] at (of) {};
\node[silentTrans, right of=mf, xshift=-10mm,original] (mc) {};
\node[silentTrans, right of=of, xshift=-10mm,original] (oc) {};
\node[place, above of=m0, yshift=-14mm, xshift=17mm] (lmo) {}; 
\node[yshift=4mm, insc,xshift=4mm] at (lmo) {$p_{\sigma_m,\sigma_o}$};
\draw[arc] (im) -- node[above, insc, near start] {$\nu_{\sigma_m}$} (mb);
\draw[arc] (io) -- node[above, insc, near start] {$\nu_{\sigma_o}$} (ob);
\draw[fatarc] (mb) -- node[above, insc] {$X_{\sigma_m}$} (cmo);
\draw[arc] (ob) -- node[below, insc] {$x_{\sigma_o}$} (cmo);
\draw[fatarc] (cmo) -- node[above, insc] {$X_{\sigma_m}$} (ma);
\draw[arc] (cmo) -- node[below, insc] {$x_{\sigma_o}$} (oa);
\draw[fatarc] (cmo) |- node[above, yshift=-0mm, insc] {$\tup{X_{\sigma_m}, x_{\sigma_o}}$} (lmo);
\draw[arc] (ma) -- node[above, insc] {$x_{\sigma_m}$} (em);
\draw[arc] (oa) -- node[above, insc] {$x_{\sigma_o}$} (eo);
\draw[arc, original] (em) -- node[above, insc,black] {$x_{\sigma_m}$} (m0);
\draw[arc, original] (eo) -- node[above, insc,black] {$x_{\sigma_o}$} (o0);
\draw[arc,original] (mf) -- node[above, insc,original] {$x_{\sigma_m}$} (mc);
\draw[arc] (lmo) -| node[above, insc, near start,xshift=4mm] {$\tup{x_{\sigma_m},x_{\sigma_o}}$} (mc);
\draw[arc,original] (of) -- node[above, insc,original] {$x_{\sigma_o}$} (oc);
\tikzstyle{group}=[red!90!black,scale=.9,yshift=-8mm]
\node[yshift=11mm, xshift=-4mm, group] at (lmo) {$P_L$};
\node[original,trans,fill=white,right of=m0,yshift=-5mm,xshift=-5mm] (t) {};
\node[original,scale=.8] at ($(t)+(-.7,0)$) {\dots};
\node[original,scale=.8] at ($(t)+(.7,0)$) {\dots};
\draw[fatarc, <->] (t) -- (lmo);
\draw[fatarc,original] ($(t)+(-.5,.2)$) -- (t);
\draw[arc,original] ($(t)+(-.5,-.2)$) -- (t);
\draw[fatarc,original] (t) -- ($(t)+(.5,.2)$);
\draw[arc,original] (t) -- ($(t)+(.5,-.2)$);
\node[group] at (io) {$T_i$};
\node[group] at (ob) {$P_b$};
\node[group,yshift=-2mm] at (cmo) {$T_L$};
\node[group] at (oa) {$P_a$};
\node[group,original] at (eo) {$T_e$};
\node[group,original] at (o0) {$P_{\tikz{\node[starttoken,original,minimum width=1.8mm] {};}}$};
\node[group,original] at (of) {$P_{\tikz{\node[endtoken,original,minimum width=1.8mm] {};}}$};
\node[group,original] at (oc) {$T_c$};
\end{tikzpicture}}\par

Before defining the new flow relations, we give an intuitive explanation considering the many-to-one relationship between wheels and a frame.
The idea is that all links for the later execution are created concurrently before the objects are placed in the initial places $\Pplay$ by $T_e$.
To that end, new emitting transitions $T_i$, instead of $T_e$, are used to create objects which are placed in $P_b$. Every link creation transition $\tlink{\sigma_m, \sigma_o} \in T_L$ is connected with inflow to the according places before link creation $P_b$, and with outflow to the respective link place $p_{\sigma_m,\sigma_o}\in P_L$ as well as places in $P_a$ after link creation.
The places in $P_a$ are connected to the former emitting transitions $T_e$.
When all transitions $T_e$ emitted the objects to the places $\Pplay$, the places in $P_L$ are thus populated with a fixed set of links.
Instead of emitting the wheels and frames directly to the $\Pplay$ places, each frame is considered jointly with a non-empty set of wheels by a transition $T_L$ to create the required links onto a place in $P_L$.
Crucially, to enforce the synchronization for a transition $t$ in $N_1$, which operates on $\sigma_m$ and $\sigma_o$, that is wheels and a frame, the respective transition in $\linkmap(N_1)$ must read and write the required links from the link places $P_L$.
at the end of the execution, any created link on $P_L$ is consumed together with the object of the many-side by the consuming transition in $T_c$, ensuring that all links are eventually consumed as well.

\noindent
Precisely, we define the flow relations as follows.
\begin{compactitem}
\item 
For a new emitting transition $t_{i,\sigma} \in T_i$, and $p \in P_b$ where $\coloring_2(p) = \tup{\sigma}$, the outflow is set as $F_{out,2}(t_{i,\sigma},p) = \tup{\nu_{\sigma}}$.
\item 
For a former emitting transition $t_{e,\sigma} \in T_e$ and $p\in \Pplay$ with $\coloring_2(p)=\tup{\sigma}$, we set
$F_{out,2}(t_{e,\sigma},p) = \tup{x_{\sigma}}$ if there is a $\sigma' \in \objtypes$ such that $(\sigma, \sigma')$ or $(\sigma', \sigma)$ in $R$, otherwise $F_{out,2}(t_{e,\sigma},p) = \tup{\nu_{\sigma}}$ as before.
\item 
For a place in $P_b$ 
which is for the one-side of the relationship, the inflow reads a single object $F_{in,2}(p_{\sigma_o,(\sigma_m, \sigma_o)},\tlink{\sigma_m, \sigma_o}) = \tup{x_{\sigma_o}}$; while for the many-side we have $F_{in,2}(p_{\sigma_m,(\sigma_m, \sigma_o)},\tlink{\sigma_m, \sigma_o}) = \tup{X_{\sigma_m}}$ reading a list of objects.
\item
Similarly, for a place in $P_a$ for the one-side of a relationship we have $F_{out,2}(\tlink{\sigma_m, \sigma_o},p_{(\sigma_m, \sigma_o),\sigma_o}) = \tup{x_{\sigma_o}}$ whereas for a place for the many-side we set $F_{out,2}(\tlink{\sigma_m, \sigma_o}, p_{(\sigma_m, \sigma_o),\sigma_m}) = \tup{X_{\sigma_m}}$.
\item 
Each link creation transition $\tlink{\sigma_m, \sigma_o}  \in T_L$ has an outflow to the respective place $p_{(\sigma_m, \sigma_o)} \in P_L$ with inscription $F_{out,2}(\tlink{\sigma_m, \sigma_o} , p_{(\sigma_m, \sigma_o)}) = \tup{X_{\sigma_m}, x_{\sigma_o}}$.
\item
Every former emitting transition $t_{e, \sigma} \in T_e$ has an inflow from all related $p \in P_a$ where $\coloring_2(p) = \tup{\sigma}$, so that $F_{in,2}(p, t_{e, \sigma}) = \tup{x_{\sigma}}$.
\item
For all link places $p_{(\sigma_m, \sigma_o)} \in P_L$ and transitions $t \in T$ such that there exist places $p_1, p_2 \in P$ with $\coloring_2(p_1) = \tup{\sigma_m}$ and $\coloring_2(p_2) = \tup{\sigma_o}$ and both $F_{in,1}(p_1,t)$ and $F_{in,1}(p_2,t)$ are defined, the transition needs to read and write the appropriate links.
If $t$ consumes a variable number of $\sigma_m$ objects, i.e. if $F_{in,1}(p_1,t)$ is a variable flow, we set $F_{in,2}(p_{(\sigma_m, \sigma_o)},t) = F_{out,2}(t, p_{(\sigma_m, \sigma_o)}) = \tup{X_{\sigma_m}, x_{\sigma_o}}$;
otherwise, $F_{in,2}(p_{(\sigma_m, \sigma_o)},t) = F_{out,2}(t, p_{(\sigma_m, \sigma_o)}) = \tup{x_{\sigma_m}, x_{\sigma_o}}$.
\item 
Finally, for all link places $p_{(\sigma_m, \sigma_o)} \in P_L$, an arc to the consuming transition $t_{c,\sigma_m} \in T_c$ consumes the links $F_{in,2}(p_{(\sigma_m, \sigma_o)},t_{c,\sigma_m}) = \tup{x_{\sigma_m}, x_{\sigma_o}}$.
\end{compactitem}
Note that for the modification of transitions from the OCPN, we assume that the OCPN is constructed from sound workflow nets. Hence, every transition has an inflow for an object type if and only if it has an outflow for it.
For example, if $p_{W,F}$ holds the tokens $(w_1,f_1), (w_2,f_1)$, then the transition \texttt{assemble\_w} can only bind the frame $f_1$ to any of the wheels $w_1$ or $w_2$ but never to $w_3$.

For all other $p\in P_1$ and $t\in T_1$
we set $F_{in,2}(p,t) = F_{in,1}(p,t)$, $F_{out,2}(t,p) = F_{out,1}(t,p)$, and $\coloring_2(p) = \coloring_1(p)$.
The new places are $P_2 = P_1 \cup P_b \cup P_L \cup P_a$, and the new transitions are $T_2 = T_1 \cup T_i \cup T_L$.
Then, the result of the mapping $\linkmap$ is the OPID 
$\linkmap(\directmap(N)) = (\objtypes_1, P_2, T_2, F_{in,2}, F_{out,2}, \coloring_2, \ell_2)$ 
where $\ell_2(t) = \tau$ for all $t \in T_i \cup T_L$ and $\ell_2(t) = \ell_1(t)$ otherwise.

The OPID $N_2$ is again interpreted as an accepting OPID where the initial marking $M_{init,2}$ and the final marking $M_{\mathit{final},2}$ require that all places be empty.
The OPID in \autoref{example:opid:sync} accepts $L_1$ because there exists a run, where the links for wheels and frames can be set as $(w_1, f_1), (w_2, f_1), (w_3, f_2), (w_4, f_2)$ in the orange link place and the tokens for handlebars and frames as $(h_1,f_1), (h_2,f_2)$ in the green link place. The transitions must bind objects correctly according to the links.
However, the execution in $L_2$ is rejected because $\tup{\m{collect},\{f_3, h_3, w_5, w_6\}}$ requires the links $(w_5,f_3)$ and $(w_6,f_3)$ but $\tup{\m{assemble\_w},\{f_4, w_6, w_7\}}$ requires the link $(w_6,f_4)$, which cannot have been created in the same execution.

\begin{example}[Synchronizing OPID for bike manufacturing]\label{example:opid:sync}The OPID is obtained from the OCPN in \autoref{example:ocpn} by applying $\directmap$ and $\linkmap$ for the stable relationships \textit{(Wheel, Frame)} (orange) and \textit{(Handlebar, Frame)} (green).

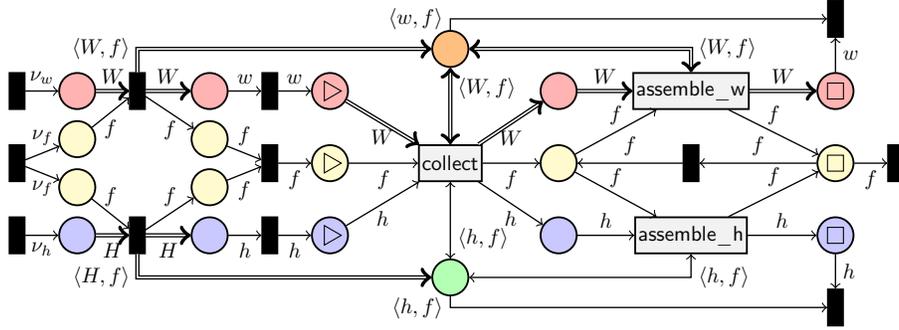
{\begin{centering}
\begin{tikzpicture}[node distance=22mm]
\node[silentTrans] (iw) {};
\node[silentTrans, below of=iw, yshift=10mm] (if) {};
\node[silentTrans, below of=if, yshift=10mm] (ih) {};
\node[wheelPlace, right of=iw, xshift=-12mm] (wb) {};
\node[framePlace, right of=if, xshift=-12mm, yshift=4mm] (fb1) {};
\node[framePlace, right of=if, xshift=-12mm, yshift=-4mm] (fb2) {};
\node[handlePlace, right of=ih, xshift=-12mm] (hb) {};
\node[silentTrans, right of=fb1, xshift=-12mm, yshift=8mm] (cwf) {};
\node[silentTrans, right of=fb2, xshift=-12mm, yshift=-8mm] (chf) {};
\node[wheelPlace, right of=wb] (wa) {};
\node[framePlace, right of=fb1] (fa1) {};
\node[framePlace, right of=fb2] (fa2) {};
\node[handlePlace, right of=hb] (ha) {};
\node[silentTrans, right of=wa, xshift=-12mm] (ew) {};
\node[wheelPlace, right of=ew, xshift=-12mm] (w0) {} node[pseudostart] at (w0) {};
\node[wheelPlace, right of=w0, xshift=16mm] (w1) {};
\node[wheelPlace, right of=w1, xshift=24mm] (w2) {} node[pseudoend] at (w2) {};
\node[silentTrans, above of=w2, yshift=-10mm] (cw) {};
\node[silentTrans, below of=ew, yshift=10mm] (ef) {};
\node[framePlace, below of=w0, yshift=10mm] (f0) {} node[pseudostart] at (f0) {};
\node[framePlace, below of=w1, yshift=10mm] (f1) {};
\node[framePlace, below of=w2, yshift=10mm] (f2) {} node[pseudoend] at (f2) {};
\node[silentTrans, right of=f2, xshift=-12mm] (cf) {};
\node[silentTrans, below of=ef, yshift=10mm] (eh) {};
\node[handlePlace, below of=f0, yshift=10mm] (h0) {} node[pseudostart] at (h0) {};
\node[handlePlace, below of=f1, yshift=10mm] (h1) {};
\node[handlePlace, below of=f2, yshift=10mm] (h2) {} node[pseudoend] at (h2) {};
\node[silentTrans, below of=h2, yshift=10mm] (ch) {};
\node[trans, right of=f0, xshift=-2mm] (cp) {$\m{collect}$};
\node[trans, right of=w1] (aw) {$\m{assemble\_w}$};
\node[silentTrans, right of=f1] (s1) {};
\node[trans, right of=h1] (ah) {$\m{assemble\_h}$};
\node[wheelFramePlace, above of=cp, yshift=-3mm] (lwf) {};
\node[handleFramePlace, below of=cp, yshift=3mm] (lhf) {};
\draw[arc] (iw) -- node[above, insc] {$\nu_w$} (wb);
\draw[arc] (if) -- node[above, insc] {$\nu_f$} (fb1);
\draw[arc] (if) -- node[below, insc] {$\nu_f$} (fb2);
\draw[arc] (ih) -- node[below, insc] {$\nu_h$} (hb);
\draw[arc] (w2) -- node[right, insc] {$w$} (cw);
\draw[arc] (f2) -- node[below, insc] {$f$} (cf);
\draw[arc] (h2) -- node[right, insc] {$h$} (ch);
\draw[fatarc] (wb) -- node[above, insc] {$W$} (cwf);
\draw[fatarc] (cwf) -- node[above, insc] {$W$} (wa);
\draw[arc] (wa) -- node[above, insc] {$w$} (ew);
\draw[arc] (ew) -- node[above, insc] {$w$} (w0);
\draw[fatarc] (w0) -- node[below, insc] {$W$} (cp);
\draw[fatarc] (cp) -- node[below, insc] {$W$} (w1);
\draw[fatarc] (w1) -- node[above, insc] {$W$} (aw);
\draw[fatarc] (aw) -- node[above, insc] {$W$} (w2);
\draw[arc] (fb1) -- node[below, insc] {$f$} (cwf);
\draw[arc] (fb2) -- node[above, insc] {$f$} (chf);
\draw[arc] (cwf) -- node[below, insc] {$f$} (fa1);
\draw[arc] (chf) -- node[above, insc] {$f$} (fa2);
\draw[arc] (fa1) -- node[above, insc] {$f$} (ef);
\draw[arc] (fa2) -- node[below, insc] {$f$} (ef);
\draw[arc] (ef) -- node[below, insc] {$f$} (f0);
\draw[arc] (f0) -- node[below, insc] {$f$} (cp);
\draw[arc] (cp) -- node[below, insc] {$f$} (f1);
\draw[arc] (f1) -- node[above, insc] {$f$} (aw);
\draw[arc] (aw) -- node[above, insc] {$f$} (f2);
\draw[arc] (f1) -- node[above, insc] {$f$} (ah);
\draw[arc] (ah) -- node[above, insc] {$f$} (f2);
\draw[arc] (f2) -- node[above, insc] {$f$} (s1);
\draw[arc] (s1) -- node[above, insc] {$f$} (f1);
\draw[fatarc] (hb) -- node[below, insc] {$H$} (chf);
\draw[fatarc] (chf) -- node[below, insc] {$H$} (ha);
\draw[arc] (ha) -- node[below, insc] {$h$} (eh);
\draw[arc] (eh) -- node[below, insc] {$h$} (h0);
\draw[arc] (h0) -- node[below, insc] {$h$} (cp);
\draw[arc] (cp) -- node[below, insc] {$h$} (h1);
\draw[arc] (h1) -- node[above, insc] {$h$} (ah);
\draw[arc] (ah) -- node[above, insc] {$h$} (h2);
\draw[fatarc] (cwf) |- node[left, insc] {$\langle W, f \rangle$} (lwf);
\draw[fatarc, <->] (lwf) -- node[right, insc, yshift=3mm] {$\langle W, f \rangle$} (cp);
\draw[arc, <->] (lhf) -- node[right, insc, yshift=-3mm] {$\langle h, f \rangle$} (cp);
\draw[fatarc, <->] (lwf) -| node[right, insc] {$\langle W, f \rangle$} (aw);
\draw[arc] (lwf) |- node[left, insc] {$\langle w, f \rangle$} (cw);
\draw[fatarc] (chf) |- node[left, insc] {$\langle H, f \rangle$} (lhf);
\draw[arc, <->] (lhf) -| node[right, insc] {$\langle h, f \rangle$} (ah);
\draw[arc] (lhf) |- node[left, insc] {$\langle h, f \rangle$} (ch);
\end{tikzpicture}
\end{centering}}
\end{example}

\subsection{Equivalence for Executions with Stable Relationships}

In~\autoref{sec:equiv}, we show that an OCPN $N$ is equivalent to its OPID representation $N_1 = \directmap(N)$ for event log replay.
The mapping $N_2 = \linkmap(N_1)$ implements the stable relationships $R$ into $N_1$. We show that a run in the OCPN $N$ is conform to $R$ if and only if there is an accepted corresponding run in $N_2$.

\begin{theorem}
\label{theorem:sync:1}
Let $N_1\,{=}\,\directmap(N)$ be an equivalent OPID for an OCPN $N$ and $N_2=$ $\linkmap(N_1)$ a synchronizing OPID. There is an accepted run of $N_1$ of the form
\[
\Mempty \goto{*} M_0 \goto{(t_1,\beta_1)} M_1 \goto{(t_2,\beta_2)} \dots \goto{(t_n,\beta_n)} M_n \goto{*} \Mempty
\]
that conforms to the stable relationships $R$ if and only if there is an accepted run of $N_2$ of the form
\[
\Mempty \goto{*} M_b' \goto{*} M_a' \goto{*} M_0' \goto{(t_1,\beta_1)} M_1' \goto{(t_2,\beta_2)} \dots \goto{(t_n,\beta_n)} M_n' \goto{*} \Mempty
\]
such that $\forall p \in P_1 \colon M_i'(p)=M_i(p)$ for all $0 \leq i \leq n$.
\end{theorem}
%
While the detailed proof for \autoref{theorem:sync:1} is spelled out in Appendix~\ref{app:b} of the paper,
%
we outline its three main arguments: (i) Any accepted run in $N_2$ is accepted in $N_1$. (ii) Any accepted run in $N_1$ conforming to $R$ is accepted by $N_2$. (iii) Any accepted run in $N_1$ not conforming to $R$ is rejected by $N_2$.
\begin{compactenum}
    \item[(i)] The emitting and link creating sequence $\Mempty \goto{*} M_b' \goto{*} M_a' \goto{*} M_0'$ results in a marking $M_0'$ that holds only tokens in the places $\Pplay \cup P_L$. An equivalent emitting sequence $\Mempty \goto{*} M_{0}$ exists for $N_1$ that produces the same objects in the initial places $\Pplay$.
    Any binding execution $(t_i,\beta_i)$ in the run of $N_2$ can be replayed in a corresponding marking of $N_1$ as it is less restrictive.
    Finally, there exists an appropriate sequence $M_n \goto{*} \Mempty$ of only consuming transitions $t \in T_c$ that removes the same objects as $M_n' \goto{*} \Mempty$.
    \item[(ii)] For the conforming run $\Mempty \goto{*} M_0$ that places objects to $\Pplay$ in $N_1$, there exists a corresponding run $\Mempty \goto{*} M_b' \goto{*} M_a' \goto{*} M_0'$ in $N_2$, which places the same objects onto the places $P_b$ in $M_b'$. Then, $M_b' \goto{*} M_a'$ creates all required links according to the binding executions $(t_i,\beta_i)$ in the run of $N_1$, such that for every object $o_m$ of type $\sigma_m$ in a relationship with an object of type $\sigma_o$, the links in $P_L$ are set for all $lo(o_m, \sigma_o)$. Next, $M_a' \goto{*} M_0'$ places the objects in $\Pplay$ such that $M_0'$ holds tokens only in places $p' \in \Pplay \cup P_L$. The created links are not changed and transitions that operate on two object types in a relationship are required to adhere to them.\\
    If the run in $N_1$ is conform, it only binds an object $o_m$ to a single object of $\sigma_o$, hence, $|lo(o_m, \sigma_o)| = 1$ and there exists a link such that $N_2$ can accept the run as well.
    Finally, for $M_n \goto{*} \Mempty$ of only consuming transitions $t \in T_c$, $M_n' \goto{*} \Mempty$ can consume the same objects including their respective links.
     \item[(iii)] By contradiction, we show that $N_2$ cannot create links that violate the conformance criteria. No object in a relationship can exist without a link. Any object $o_m$ with $ot(o_m) = \sigma_m$ will receive exactly one link to an object of type $\sigma_o$ and not more. As the bindings of any run in $N_2$ are restricted by the links, a non-conforming run of $N_1$ has no counterpart in $N_2$.
\end{compactenum}
By expanding the mapping $\directmap$ with stable relationships in $\linkmap$ for a given OCPN and its relationships $R$, we can generate an OPID $\linkmap(\directmap(N))$ that accepts an execution in $N$ if and only if it conforms to $R$.

\section{Discussion and Related Work}\label{sec:related_work}

In the following, we discuss the presence of stable many-to-one relationships in existing benchmark OCELs, and we outline a prototype implementation of our mapping. 
Both artifacts are available online\footnote{\url{https://github.com/AnjoSs/To-bind-or-not-to-bind}}.
Finally, we comment on the implications for the research areas of conformance checking and formal analysis.

\subsection{Stable Relationships in OCEL}
Stable relationships can be discovered at the type level from OCELs, as described in the TOTeM tool~\cite{lissTOTeMTemporalObject2024}. To substantiate our intuition that, in many contexts, the executions of object-centric processes indeed have stable relationships, we applied an according analysis to existing benchmark logs.
The result in \autoref{tab:logs} shows that all investigated logs contain some stable many-to-one relationships.
The full data is available from the repository.

The stability of object relationships is a typical (often implicit) assumption in the literature. In the widely adopted OCEL 2.0 standard~\cite{bertiOCELObjectCentricEvent2024}, links are recorded without a notion of time, assuming global existence and stability. For discovered models, such as OCPNs~\cite{vanderaalstDiscoveringObjectcentricPetri2020}, class diagrams~\cite{lissTOTeMTemporalObject2024}, and even simulation models~\cite{knoppDiscoveringObjectCentricProcess2023}, this yields the same notion of stable relationship introduced here.

The intended stable relationships may be violated in OCELs. When investigating noise, that is, infrequent outliers, it can be seen that some logs contain a few violations of the intended stable many-to-one relationships. In \autoref{tab:logs}, we report how the number of stable relationships changes when considering noise. 
For example, LLM-based simulations for LRMS -- O2C contain 20 more stable relationships when only focusing on the 90\% most frequent observations.
This relates to the approach in~\cite{xiuQuantifyingConformanceObjectCentric2024}, where this type of violation is identified in a real-world use case.
Our mapping to OPIDs can assist in identifying exact executions that violate stable relationships.
\begin{table}[t]
    \caption{OCELs with their number of object types, all relationships (including one-to-one), stable many-to-one relationships with, and without filtering noise.}
    \vspace*{-2mm}
    \centering
   \begin{tabular}{|l|c|c|c|c|}
       \hline
       \textbf{OCEL}        & {\#obj. types} & \#relationships & {\#m2o(noise=0.1)} & {\#m2o(noise=0.0)} \\
       \hline
       Logistics   & 7 & 9 & 6 (4+2*1) & 6 (4+2*1)\\
       \hline
       Order Management   & 6 & 13 & 5 (5+0) & 5 (5+0)  \\
       \hline
       Procure to Pay   & 7 & 9 & 10 (6+2*2) & 10 (6+2*2) \\
       \hline
       Hinge Production   & 12 & 39 & 35 (13+2*11) & 35 (13+2*11) \\
       \hline
       Age of Empires   & 30 & 135 & \textbf{87} (77+2*5) & \textbf{67} (57+2*5) \\
       \hline
       LRMS -- O2C   & 9 & 11 & \textbf{18} (2+2*8) & \textbf{16} (4+2*6)  \\
       \hline
       LRMS -- P2P   & 6 & 15 & \textbf{11}  (3+2*4) & \textbf{0} (0+0) \\
       \hline
       LRMS -- Hiring   & 6 & 15 & 6 (6+0) & 6 (6+0) \\
       \hline
       LRMS -- Hospital   & 8 & 18 & 6 (6+0) & 6 (6+0) \\
       \hline
    \end{tabular}
    \label{tab:logs}
    \vspace{-7mm}
\end{table}

However, our identified notion of stable relationships yield some limitations.
Our definition of stable many-to-one relationships only considers existential relations of the type \texttt{(1..*,1..1)}. Non-existential relationships (e.g. \texttt{(0..*,0..1)}) are not present in the existing OCELs but should be considered in future investigations as well.
In general, the need to enrich object-centric event data with data modeling concepts such as non-rigid relationships is being increasingly recognized.
Goossens et al.~\cite{goossensObjectCentricEventLogs2024} point out that from a database perspective, links between objects might indeed change over time. 
It is not possible to differentiate between a non-rigid many-to-one relationship and an intended many-to-many relationship without domain knowledge.
Therefore, our approach emphasizes an explicit declaration of object types and relationships, e.g., as data models.
In addition,~\cite{swevelsImplementingObjectCentricEvent2024} highlights the need for richer representations of object-centric event data, where the effects of events are explicitly tracked, allowing identification of time intervals within which a relationship exists.
An ongoing effort aims to consolidate these different considerations in a single standard for representing object-centric event data \cite{FMLA24}, motivating further research on the matter.

\subsection{Process Discovery and Prototypical Implementation}
The original discovery of OCPNs~\cite{vanderaalstDiscoveringObjectcentricPetri2020} is proposed as a divide-and-conquer approach in which each object type defines an individual workflow net~\cite{vanderaalstDiscoveringObjectcentricPetri2020}.
The proposed mapping from OCPNs to OPIDs is based on this specific characteristic structure of OCPNs. Therefore, our mapping rules and formal proofs are limited to such OCPNs.
OCPN discovery~\cite{vanderaalstDiscoveringObjectcentricPetri2020} has recently been extended in~\cite{vandettenDiscoveringCompactLive2024}, which proposes the discovery of object-centric process trees. From these process trees, identifier-sound~\cite{vanderwerfCorrectnessNotionsPetri2024} OCPNs can be generated, but without synchronization semantics. Analog to ~\cite{vanderaalstDiscoveringObjectcentricPetri2020}, the resulting Petri nets are constructed as combinations of workflow nets for each object type.
In a second extension, ~\cite{vandettenObjectSynchronizationsSpecializations2024} discovers so-called silent objects from an OCEL, which reify and synchronize arbitrary sets of objects.
This extension adds places for the discovered silent object types and enforces synchronization of the observed co-occuring objects. 
However, the discovered Petri nets are also limited to replaying the behavior of an OCEL. Still, the concept of silent objects may provide insights on how to extend our investigation to many-to-many relationships.

However, the discovery of OPIDs has not yet been tackled.
To show how discovered OCPNs can be directly mapped to OPIDs, we provide a prototypical implementation of our proposed mappings.
Our prototype implements both transformations $\directmap$ and $\linkmap$ in a Python command-line script. The script takes as input an OCPN in json format, and for $\linkmap$ the desired stable many-to-one relationships $R$. The resulting OPID is output in a pnml format, which can be used as input for conformance checking~\cite{gianolaObjectCentricConformanceAlignments2024}, and in the dot format for visualization.
The code, as well as the results obtained by applying it to a set of OCPNs collected in earlier work~\cite{ocpn-visualizer}, can be obtained from the repository.

\subsection{Conformance Checking and Formal Analysis}
OCPNs have been studied for conformance checking.
Fitness and precision metrics have been introduced to quantify the suitability of an OCPN to replay a log~\cite{adamsPrecisionFitnessObjectCentric2021}, and techniques have been studied to calculate alignments for OCPNs~\cite{LiAA23}. However, both focus only on individual objects and not on their mutual links.
For varying links, these approaches are too liberal in identifying deviations~\cite{gianolaObjectCentricConformanceAlignments2024}, calling for richer formalisms such as PNIDs or OPIDs, which are computationally more demanding~\cite{gianolaObjectCentricConformanceAlignments2024}.
Still, with \autoref{theorem:sync:1} we show that conformance checking with OCPNs~\cite{adamsPrecisionFitnessObjectCentric2021, LiAA23} is appropriate if the log satisfies the intended stable relationships. In practice, this check through the discovery of data models~\cite{lissTOTeMTemporalObject2024} can be more performant than the computation of synchronization alignments~\cite{gianolaObjectCentricConformanceAlignments2024}.
If the log violates stable relationships, OPIDs can be used to compute conformance alignments, including specific penalties to quantify the severity of the violations~\cite{gianolaObjectCentricConformanceAlignments2024}.
In the future, the fitness and precision metrics should be extended to respect the object relationships defined in OPIDs.
Other process models have also been considered for object-centric conformance checking, that is, object-centric directly follows graphs~\cite{parkConformanceCheckingPerformance2024a} and object-centric behavioral constraints~\cite{xiuDiagnosingConformanceObjectCentric2023}.

Different formal modeling approaches have been proposed for object-centric processes, such as variants of Petri nets with identifiers (PNIDs) \cite{ghilardiPetriNetbasedObjectcentric2022a,vanderwerfCorrectnessNotionsPetri2024}, or synchronous proclets \cite{fahlandDescribingBehaviorProcesses2019}.
These two formalisms provide extended orthogonal features when related to OCPNs. On the one hand, PNIDs provide explicit modeling mechanisms for creating, updating, and removing objects and relationships, while lacking variable arcs to evolving multiple objects in a single transition. On the other hand, proclets extend the variable-arc feature of OCPNs, by supporting different forms of synchronization where many objects of a certain type synchronously move together with a single, parent object; however, they lack support for the explicit manipulation of objects and relationships, as supported in PNIDs.
In a recent effort, \emph{object-centric PNIDs} (OPIDs \cite{gianolaObjectCentricConformanceAlignments2024}) have emerged to provide an integrated formalism that unifies the distinctive features of OCPNs with variable arcs, PNIDs and synchronous proclets. Thus, OPID is the natural candidate for enriching an OCPN with the desired synchronization behavior.

As we show, OPIDs with synchronization can be mapped directly from discovered OCPNs.
As OPIDs can emit novel objects, they can serve not only as trace recognizers but also as trace generators. Therefore, our mapping serves as a missing link between object-centric process discovery and their formal analysis and execution.
The relationship between the established approaches of Proclets and PNIDs, the novel OPIDs and discovered process models needs to be studied in the future, where our mapping can now also serve as the missing link.
For example, van der Werf et al.~\cite{vanderwerfCorrectnessNotionsPetri2024} define correctness notions for PNIDs, which can now be applied to discovered process models through our mapping.
Finally, OPIDs can also serve as a goal formalism for object-centric process modeling approaches, enabling the comparison of model semantics and discovery results. 
\section{Conclusion}\label{sec:conclusion}

This paper shows how discovered object-centric process models can be enhanced with explicit data modeling concepts to ensure precision and trust in object relationships.
Our contribution rigorously closes the gap between the two representative formal models for object-centric processes: OCPNs---an underfitting model with efficient discovery and available conformance checking techniques, and OPIDs---a more sophisticated model for which discovery techniques are still missing and conformance checking techniques are computationally more demanding.
Our findings indicate that when focusing on executions where many-to-one relationships are stable, that is, once a link is established it is rigidly maintained over time, these two models coincide.
The implications of this result can be read in two ways.
(i) While OCPNs are, in general, underfitting, they become precise enough when describing and prescribing executions that establish stable many-to-one relationships. 
(ii) When restricting attention to this type of executions, discovery and conformance checking for OPIDs can be directly tackled through techniques natively defined for OCPNs. 
In the future, stability in many-to-many relations in event logs remains to be studied in detail.
These results can serve as a promising link between object-centric process discovery and formal analysis.

\begin{credits}
\subsubsection{\ackname}
M. Montali was partially supported by the NextGenerationEU FAIR PE0000013 project MAIPM (CUP C63C22000770006) and the PRIN MIUR project PINPOINT Prot. 2020FNEB27. S. Winkler was partially supported by the UNIBZ project TEKE. A. Gianola was partly supported by Portuguese national funds through Fundação para a Ciência e a Tecnologia, I.P. (FCT), under projects UIDB/50021/2020 (DOI: 10.54499/UIDB/50021/2020). This work was partially supported by the ‘OptiGov’ project, with ref. n. 2024.07385.IACDC (DOI: 10.54499/2024.07385.IACDC), fully funded by the ‘Plano de Recuperação e Resiliência’(PRR) under the investment ‘RE-C05-i08 - Ciência Mais Digital’, measure ‘RE-C05-i08.m04’ (in accordance with the FCT Notice No. 04/C05-i08/2024), framed within the financing agreement signed between the ‘Estrutura de Missão Recuperar Portugal’ (EMRP) and FCT as an intermediary beneficiary.

% ...
\end{credits}

\bibliographystyle{splncs04}
\bibliography{paper.bib}

\newpage
\appendix
\section{Proof for Equivalent OPIDs}
\label{app:a}
\setcounter{theorem}{0}

To prove the language equivalence of $N$ and $\directmap(N)$ in general, we first define some auxiliary notions.

Given a marking $M\colon Q \to \{0,1\}$ of the OCPN $N$, we define the corresponding marking $\directmap(M)$ of the OPID $\directmap(N)$ as $\directmap(M)(p) := \{ \langle o \rangle \mid o \in \objects \wedge M(p,o) > 0\}$.

Given a binding $b$ for an OCPN $N$, we define a corresponding binding $\beta:=\directmap(b)$ for the OPID $\directmap(N)$ as follows:
for all $\sigma\in\dom(b)$, we set $\beta(v_\sigma)=b(\sigma)$ if $|b(\sigma)| = 1$, and  $\beta(V_\sigma)=[b(\sigma)]$ otherwise, where
$[b(\sigma)]$ denotes a list containing the elements of the set $b(\sigma)$ in an arbitrary order.

Moreover, we call a sequence of transition firings $M \goto{*} M'$ in $\directmap(N)$ an \emph{emitting sequence} if it uses only emitting transitions in $T_e$ and a \emph{collapsing sequence} if it uses only consuming transitions in $T_c$.

We can then show the following equivalence statement:

\begin{theorem}
\label{theorem:equal}
Let $O$ be a set of objects, $N$ be an accepting OCPN, $N' = \directmap(N)$, and $\Mempty$ denote the empty marking of $N'$.

There exists a run of $N$:
\begin{equation}
\label{eq:OCPN:run:ex}
M_0 \goto{(t_1,b_1)} M_1 \goto{(t_2,b_2)} \dots \goto{(t_n,b_n)} M_n
\end{equation}
with $M_0{=}\Minit^O$ and $M_n{=}\Mfinal^O$
if and only if there is a run of $N'$ over the set of objects $O$:
\begin{equation}
\label{eq:OPID:run:ex}
\Mempty \goto{*} M_0' \goto{(t_1,\beta_1)} M_1' \goto{(t_2,\beta_2)} \dots \goto{(t_n,\beta_n)} M_n' \goto{*} \Mempty
\end{equation}
such that $M_i'=\directmap(M_i)$ for all $0 \leq i \leq n$ and $\beta_i = \directmap(b_i)$ for all $1\leq i \leq n$,
such that $\Mempty \goto{*} M_0'$ is an emitting sequence, $M_n' \goto{*} \Mempty$ a collapsing sequence, and no emitting or collapsing transitions occur in between.
\end{theorem}

\begin{proof}
\begin{compactenum}
\item[(1)]
First, by definition of $T_e$ in $N'$, there is an emitting sequence $\Mempty \goto{*} M_0'$ since $M_0'=\directmap(M_0)$ has tokens only in places that are initial in $N$. Moreover, by definition of $T_c$ there is a collapsing sequence $M_n' \goto{*} \Mempty$ because $M_n'=\directmap(M_n)$ has tokens only in places that are final in $N$.

Next, we show the following ($\star$): 
for all $1\leq i \leq n$, 
a transition firing $M_{i-1} \goto{(t_i,b_i)} M_{i}$ implies
$M_{i-1}' \goto{(t_i,\beta_i)} M_{i}'$.
Note first that the binding $\beta_i=\directmap(b_i)$ is type-preserving and well-defined since this also holds for $b_i$.
All inscriptions in $\directmap(N)$ are singleton tuples, 
so the extension $\vec \beta_i$ of $\beta_i$ to inscriptions is given by 
$\vec \beta_i(\tup{x_\sigma}) = \{\tup{\beta_i(x_\sigma)}\}$ 
for variables $x_\sigma\in \nuvarset \cup \varset$, and by 
$\vec \beta_i(\tup{X_\sigma}) = \{\tup{a} \mid a\text{ occurs in } \beta_i(X_\sigma)\} $ 
for all $X_\sigma\in \listvarset$.
In both cases, we have 
$\vec \beta_i(x_\sigma) = \vec \beta_i(X_\sigma) = \{\tup{a} \mid a \in b_i(\sigma)\}$.

Now, as $t_i$ is enabled with $b_i$ in marking $M_{i-1}$ of $N$, 
$\Cons(t_i, b_i) \leq M_{i-1}$ holds for 
$\Cons(t_i, b_i) = \{(p, o)\mid p \in \pre{t_i} \text{ and } o \in b_i(pt(p))\}$.
We have 
$\vec \beta_i(\inflow(p,t_i))=\{ \tup{o} \mid o \in b_i(pt(p))\}$, 
and therefore
$\vec \beta_i(\inflow(p,t_i)) \subseteq M_{i-1}'(p)$ 
for all $p \in \pre{t_i}$, so that
$t_i$ is enabled in marking $\directmap(M_{i-1})$ of $N'$.
The resulting marking in $N$ is 
$M_{i} = (M_{i-1} \setminus \Cons(t_i, b_i)) \cup \Prod(t, b)$, for
$\Prod(t, b) = \{(p, o)\mid p \in \post t \text{ and } o \in b(pt(p))\}$.
The equivalent marking in $N'$ is given by 
$M_i'(p)=M_{i-1}'(p) \setminus \vec \beta_i(\inflow(p,t_i)) $ 
for all $p \in \pre{t_i}\setminus \post {t_i}$,
$M_i'(p)=M_{i-1}'(p) \cup \vec \beta_i(\outflow(p,t_i))$ 
for all $p \in {\post t_i}\setminus {\pre t_i}$, and 
$M_i'(p)=M_{i-1}'(p) \setminus \vec \beta_i(\inflow(p,t_i)) \cup \vec \beta_i(\outflow(p,t_i))$ 
for all $p \in \pre t_i\cap \post t_i$. 
We have 
$\vec \beta_i(\inflow(p,t_i))=\vec \beta_i(\outflow(p,t_i))=\{ \tup{o} \mid o \in b_i(pt(p))\}$, so $M_i' = \directmap(M_i)$. Thus $(\star)$ holds.

A simple inductive argument using $(\star)$, in addition to the observations about emitting and collapsing sequences, shows shows that the run \eqref{eq:OPID:run:ex} is possible.
\item[(2)]
First, since $\Mempty \goto{*} M_0'$ is an emitting sequence from the initial marking, $M_0'$ has tokens only in places that are initial in $N$. So there is a marking $M_0$ of $N$ that is a valid initial marking of $N$ and satisfies $M_0' = \directmap(M_0)$.

Next, we show the following ($\star$): 
for all $1\leq i \leq n$, given $M_{i-1}$ which satisfies $M_{i-1}'=\directmap(M_{i-1})$, there are a marking $M_i$ and a binding $b_i$ that satisfy
$M_{i-1} \goto{(t_i,b_i)} M_{i}$, $M_{i}'=\directmap(M_{i})$, and $\beta_i=\directmap(b_i)$. Since $\directmap$ is injective on markings and bindings, $M_i$ and $b_i$ can actually be obtained from $M_i'$ and $\beta_i$ by applying the inverse of $\directmap$.
As $t_i$ is enabled in $M_{i-1}'$, it holds that
$\vec \beta_i(\inflow(p,t_i)) \subseteq M_{i-1}'(p)$ for all $p\in \pre{t_i}$.
Since $\vec \beta_i(\inflow(p,t_i))=\{ \tup{o} \mid o \in b_i(pt(p))\}$ (see Item (1)), 
$\Cons(t_i, b_i) \leq M_{i-1}$ holds, so $t_i$ is enabled in $M_{i-1}$.
The resulting marking $M_i'$ in $N'$ satisfies 
$M_i'(p)=M_{i-1}'(p) \setminus \vec \beta_i(\inflow(p,t_i)) $ 
for all $p \in \pre{t_i}\setminus \post {t_i}$,
$M_i'(p)=M_{i-1}'(p) \cup \vec \beta_i(\outflow(p,t_i))$ 
for all $p \in {\post t_i}\setminus {\pre t_i}$, and 
$M_i'(p)=M_{i-1}'(p) \setminus \vec \beta_i(\inflow(p,t_i)) \cup \vec \beta_i(\outflow(p,t_i))$ 
for all $p \in \pre t_i\cap \post t_i$. 
The resulting marking $M_i$ in $N$ satisfies 
$M_{i} = (M_{i-1} \setminus \Cons(t_i, b_i)) \cup \Prod(t, b)$.
Since
$\vec \beta_i(\inflow(p,t_i))=\vec \beta_i(\outflow(p,t_i))=\{ \tup{o} \mid o \in b_i(pt(p))\}$, it holds that $M_i' = \directmap(M_i)$.
Thus $(\star)$ holds and an inductive argument shows that the run \eqref{eq:OCPN:run:ex} is possible.

Finally, since $M_n' \goto{*} \Mempty$ is a collapsing sequence to an empty marking, $M_n'=\directmap(M_n)$ has tokens only in places that are final in $N$, so $\Mfinal$ is a suitable final marking for $N$.
\qed
\end{compactenum}
\end{proof}

\section{Proof for Synchronizing OPIDs}
\label{app:b}

In Appendix~\ref{app:a}, the equivalence of an OCPN $N$ and its OPID representation $N_1 = \linkmap(N)$ for replaying an event log is shown.
The mapping $\linkmap$ implements the stable relationships $R$ into $N_1$.
We prove the correct mapping by showing three things: First, any accepted run in $N_2$ is also accepted in $N_1$. Second, any accepted run in $N_1$, which is conform to the stable relationships $R$, is accepted by $N_2$. Third, any accepted run in $N_1$, which that does not conform to the stable relationships $R$, is rejected by $N_2$.

\begin{theorem}
\label{theorem:sync}
Let $N_1 = \directmap(N) = (\Sigma_1, P_1, T_1, F_{in,1}, F_{out,1}, \coloring_1, \ell_1)$ be the equivalent OPID for the accepting OCPN $N$ over the set of objects $\objects$. 
Let $N_2 = \linkmap(N_1) = (\Sigma_2, P_2, T_2, F_{in,2}, F_{out,2}, \coloring_2, \ell_2)$ for the set $R$ of stable object type relationships, and $\Mempty$ denote the empty marking of $N_1$ and $N_2$.
\begin{compactenum}
\item[(1)] For any accepted run of $N_2$:
\[
\Mempty \goto{*} M_0' \goto{(t_1,\beta_1)} M_1' \goto{(t_2,\beta_2)} \dots \goto{(t_n,\beta_n)} M_n' \goto{*} \Mempty
\]
such that $\Mempty \goto{*} M_0'$ is an emitting and link creation sequence of only transitions in $t \in T_i \cup T_L \cup T_e$, $M_n' \goto{*} \Mempty$ is a collapsing sequence of only transitions $t \in T_c$, and no transitions $t \in T_i \cup T_L \cup T_e \cup T_c$ occur in between,
there is an accepted run of $N_1$ of the form
\begin{equation}
\label{eq:OPID2:run:ex}
\Mempty \goto{*} M_0 \goto{(t_1,\beta_1)} M_1 \goto{(t_2,\beta_2)} \dots \goto{(t_n,\beta_n)} M_n \goto{*} \Mempty
\end{equation}
such that $\forall p \in P_1 \colon M_i'(p)=M_i(p)$ for all $0 \leq i \leq n$. 
\item[(2)] For an accepted run of $N_1$ that respects the 
many-to-one relationships $R$:
\[
\Mempty \goto{*} M_0 \goto{(t_1,\beta_1)} M_1 \goto{(t_2,\beta_2)} \dots \goto{(t_n,\beta_n)} M_n \goto{*} \Mempty
\]
there is an accepted run of $N_2$ of the form
\begin{equation}
\label{eq:Sync:run:ex}
\Mempty \goto{*} M_b' \goto{*} M_a' \goto{*} M_0' \goto{(t_1,\beta_1)} M_1' \goto{(t_2,\beta_2)} \dots \goto{(t_n,\beta_n)} M_n' \goto{*} \Mempty
\end{equation}
such that $\forall p \in P_1 \colon M_i'(p)=M_i(p)$ for all $0 \leq i \leq n$.
\end{compactenum}
\item[(3)] For an accepted run of $N_1$ that does not respect the 
stable relationships $R$:
\[
\Mempty \goto{*} M_0 \goto{(t_1,\beta_1)} M_1 \goto{(t_2,\beta_2)} \dots \goto{(t_n,\beta_n)} M_n \goto{*} \Mempty
\]
there can be no accepted run of $N_2$ of the form
\begin{equation}
\label{eq:Sync:rej:ex}
\Mempty \goto{*} M_0' \goto{(t_1,\beta_1)} M_1' \goto{(t_2,\beta_2)} \dots \goto{(t_n,\beta_n)} M_n' \goto{*} \Mempty
\end{equation}
such that $\forall p \in P_1 \colon M_i'(p)=M_i(p)$ for all $0 \leq i \leq n$.
\end{theorem}

\begin{proof}
\begin{compactenum}
\item[(1)] 
$\Mempty \goto{*} M_0'$ in $N_2$ produces a set of objects through $T_i$, creates links
through $T_L$, and the transitions $t \in T_e$ populate the play places $p \in \Pplay$.
In $M_0'$, only the places $\Pplay$ and $P_L$ hold tokens, i.e. $M_0'(p) = \emptyset$ for all $p \in P_2 \setminus (\Pplay \cup P_L)$, because no transition $t \in T \setminus (T_e \cup T_L$ can remove the remaining tokens from any place $p \in P_b \cup P_a$, as for all those places and transitions it holds that $\inflow(p,t) = \bot$, and the execution would not be accepted for the final marking $\Mempty$.

For $N_1$, there exists a sequence $\Mempty \goto{*} M_0$ that can generate the same marking for the play places $p \in \Pplay \subset P_1$ by firing the appropriate emitting transition for every object.
Indeed, for all objects and places $p' \in \Pplay \subset P_2$, where $o \in M_0'(p')$ and $ot(o) = \sigma$, there exists an emitting transition $t \in T_e \subset T_1$ with an outflow to the according place $p \in \Pplay \subset P_1$ with $\coloring_2(p) = \sigma$ where $F_{out,1}(t,p) = \tup{\nu_\sigma}$. Therefore, the emitting transition can produce $o$ onto $p$.
Hence, given the marking $M_0'$ in $N_2$, $N_1$ can produce an equivalent marking $M_0$ for all places $p \in P_1$.

Next we show the following ($\star$):
For all $1 \leq i \leq n$, a transition firing $M_{i-1}' \goto{(t_i,\beta_i')} M_{i}'$ implies the existence of a firing
$M_{i-1} \goto{(t_i,\beta_i)} M_{i}$ where for all $p \in P_1$ it holds that $M_{i-1}(p) = M_{i-1}'(p)$.
For all transitions $t' \in T_1 \subset T_2$ of $N_2$ and $t \in T_1$ of $N_1$ and places $p' \in P_1 \subset P_2$ of $N_2$ and $p \in P_1$ of $N_1$, the inflow and outflow are the same for both nets $F_{in,1}(t,p) = F_{in,2}(t',p')$ and $F_{out,1}(t,p) = F_{out,2}(t',p')$.
Hence if the transition $(t_i,\beta_i)$ is enabled in marking $M_{i-1}'$ of $N_2$ then it is is also enabled in $M_{i-1}$ of $N_1$, and the resulting markings satisfy $M_{i}(p) = M_{i}'(p)$ for all places $p \in P_1$.

In $M_n'$, only the places $\Pstop$ and $P_L$ hold tokens, so $M_n'(p) = \emptyset$ for all $p \in P_2 \setminus (\Pstop \cup P_L)$. Otherwise, no transition $t \in T_c \subset T_2$ can remove the remaining tokens from any place $p$, as for all those places and transitions it holds that $\inflow(p,t) = \bot$, and the execution would not be accepted for the final marking $\Mempty$.
Consider the collapsing sequence $M_n' \goto{*} \Mempty$ in $N_2$ of only transitions $t' \in T_c$ and a marking $M_n$ in $N_1$, where $M_n(p) = M_n'(p)$ for all places $p \in P_1$.
There exists a collapsing sequence in $N_1$ consisting of only consuming transitions $t \in T_c$, resulting in $\Mempty$ as $F_{in,1}(p,t) = \tup{x_{\coloring_2(p)}}$ and $F_{out,1}(t,p') = \bot$ for any $p' \in P_1$.

A simple inductive argument using ($\star$), and the observations about the emitting, link creating, and the collapsing sequences, show that the run \eqref{eq:OPID2:run:ex} is possible.

\item[(2)] 
\label{theorem2:item2}
For the emitting sequence $\Mempty \goto{*} M_0$ in $N_1$ and the set of relations $R$, there is an emitting sequence $\Mempty \goto{*} M_b' \goto{*} M_a' \goto{*} M_0'$ in $N_2$ that sets the links correctly, as follows.
Let $S \subseteq \Sigma$ be the set of object types that occur in $R$.
\begin{compactitem}
\item
The sequence $\Mempty \goto{*} M_b'$ creates objects whose type occurs in $R$. To that end, it uses transitions $t_i \in T_i$ to obtain a marking $M_b'$ where for all places $p_b \in P_b$ with $\coloring_2(p_b) =\tup{\sigma}$ for $\sigma \in S$,
and the unique place $p\in \Pplay$ such that $\coloring_2(p_b) = \coloring_2(p)$, 
it holds that $M_b(p_b) = M_0(p)$ because $F_{out,2}(t_i,p_b) = \tup{\nu_{\sigma}}$.
\item
The sequence $M_b' \goto{*} M_a'$ uses only transitions $\tlink{\sigma_m, \sigma_o} \in T_L$, with the aim to populate $P_L$ with the required links.
For every relationship $(\sigma_m, \sigma_o) \in R$, there exist two places $p_{\sigma_m, (\sigma_m, \sigma_o)}, p_{\sigma_o, (\sigma_m, \sigma_o)} \in P_b$, one transition $t_{(\sigma_m, \sigma_o)} \in T_L$, two places $p_{(\sigma_m, \sigma_o), \sigma_m}, p_{(\sigma_m, \sigma_o), \sigma_o} \in P_a$, and one place $p_{(\sigma_m, \sigma_o)} \in P_L$
that satisfy $F_{in,2}(p_{\sigma_m, (\sigma_m, \sigma_o)}, t_{(\sigma_m, \sigma_o)}) = F_{out,2}(t_{(\sigma_m, \sigma_o)}, p_{\sigma_m, (\sigma_m, \sigma_o)}) = 
\tup{X_{\sigma_m}}$ (resp. $F_{in,2}(p_{\sigma_o, (\sigma_m, \sigma_o)}, t_{(\sigma_m, \sigma_o)}) = F_{out,2}(t_{(\sigma_m, \sigma_o)}, p_{\sigma_o, (\sigma_m, \sigma_o)}) = \tup{x_{\sigma_o}}$), as well as $F_{out,2}(\tlink{\sigma_m, \sigma_o}, p_{(\sigma_m, \sigma_o)}) = \tup{X_{\sigma_m}, x_{\sigma_o}}$.
For every object $o$, if $ot(o) = \sigma_m$ (resp. $ot(o) = \sigma_o$), then $\tup{o} \in M_b'(p_{\sigma_m, (\sigma_m, \sigma_o)})$ (resp. $\tup{o} \in M_b'(p_{\sigma_m, (\sigma_m, \sigma_o)})$).

For every transition $\tlink{\sigma_m, \sigma_o} \in T_L$ and every object $o_o \in \objects$ with $ot(o_o) = \sigma_o$, a binding $\beta_{o_o,\sigma_m}$ can be selected that creates all links of this type that are required later. 
Let $O_m$ be the set of all objects of type $\sigma_m$ that occur with $o_o$ in the run of $N_1$, i.e., all objects $o_m \in \objects$ such that $ot(o_m)=\sigma_m$ and $\{o_o,o_m\} \subseteq codom(\beta_i)$.
Then $\beta_{o_o,\sigma_m}(x_{\sigma_o})=o_o$ and
$\beta_{o_o,\sigma_m}(X_{\sigma_m})=[O_m]$.
As all required objects are present in the according places of $P_b$, $(\tlink{\sigma_m, \sigma_o}, \beta_{o_o,\sigma_m})$ can fire. 
Every object $o_o$ and $o_m \in O_m$ is only considered once by a transition $\tlink{\sigma_m, \sigma_o}$, as $F_{out,2}(\tlink{\sigma_m, \sigma_o}, p_{\sigma_m,(\sigma_m, \sigma_o)}) = F_{out,2}(\tlink{\sigma_m, \sigma_o}, p_{\sigma_o, (\sigma_m, \sigma_o)}) = \bot$.
Because the run in $N_1$ satisfies the relationship $(\sigma_m, \sigma_o) \in R$, for all bindings $\beta_i$ it holds that if $o_m \in \beta_i(x_{\sigma_m})$ and an object $o_o' \in \beta_i(x_{\sigma_m})$ of type $ot(o_o) = \sigma_o$, then $o_o' = o_o$. 
Therefore, for all $p_{(\sigma_m, \sigma_o)} \in P_L$ it holds that for every pair of links $(o_m, o_o),(o_m', o_o') \in M_a'(p_{(\sigma_m, \sigma_o)})$ if $o_m = o_m'$ then $o_o = o_o'$. All links in the places $P_L$ are set correctly for the run in $N_1$.

For all transitions and bindings $(t_i,\beta_i)$ in the run, the marking of all $p_{(\sigma_m, \sigma_o)} \in P_L$ is not changed as $F_{in,2}(p_{(\sigma_m, \sigma_o)},t) = F_{out,2}(t, p_{(\sigma_m, \sigma_o)})$ for all $t \in T \setminus (T_L \cup T_c)$. All relationships are stable.

\item
The sequence $M_a' \goto{*} M_0'$ fires transitions to populate $\Pplay$ for every object $o$ in the run: Let $ot(o)=\sigma$.
First, if $\sigma \not \in S$, there is a flow from the according $t_\sigma \in T_e$ to a place $p \in \Pplay$ with $F_{out,2}(t_\sigma, p) = \tup{\nu_{\sigma}}$.
Second, if $ot(o) \in S$, there is a place
$p_a \in P_a$ with $\coloring_2(p_a) = \tup{\sigma}$ that 
holds a token $\tup{o}$. Hence, a binding $\beta(x_{\sigma}) = o$ exists for all $F_{in,2}(p_a, t_e) = \tup{x_{\sigma}}$, resulting in a marking where $\tup{o} \in M_0'(p)$ because $F_{out,2}(t_e, p) = \tup{x_{\sigma}}$.
\end{compactitem}

Next we show the following ($\star$):
For all $1 \leq i \leq n$, a transition firing $M_{i-1} \goto{(t_i,\beta_i)} M_{i}$ where for all $p \in P_1$ it holds that $M_{i-1}(p) = M_{i-1}'(p)$, implies the existence of a firing $M_{i-1}' \goto{(t_i,\beta_i)} M_{i}'$, such that that $M_{i}(p) = M_{i}'(p)$.
For all transitions $t \in T_1 \setminus (T_e \cup T_c) = T_2 \setminus (T_i \cup T_L \cup T_e \cup T_c)$ and all places $p \in P_2 \setminus (P_L \cup P_b \cup P_a) = P_1$, it holds that all arc inscriptions are singleton tuples such that $F_{out,2}(t,p) = F_{out,1}(t,p)$, $F_{in,2}(p,t) = F_{in,1}(p,t)$.
If the binding $\beta_i$ considers objects $o_m \in O_m$ with $ot(o_m) = \sigma_m$ and $o_o$ with $ot(o_o) = \sigma_o$ of two types that are in a relationship $(\sigma_m, \sigma_o) \in R$, there is a place $p_{\sigma_m, \sigma_o} \in P_L \subset P_2$ where $O_m \times \{o_o\} \subseteq M_0'(p_{\sigma_m, \sigma_o})$.
Because the relationship is stable, we also have
$O_m \times \{o_o\} \subseteq M_{i-1}'(p_{\sigma_m, \sigma_o})$.
Therefore, for the binding $\beta_i$ from the run of $N_1$ can also be used in the run of $N_2$.
As inflow and outflow are equal for all $t_i$ and places $p \in (P_1 \cup P_L) \subset P_2$, for the resulting marking satisfies $M_{i}(p) = M_{i}'(p)$.

In $M_n$, only the places $\Pstop$ hold tokens, such that $M_n(p) = \emptyset$ for all $p \in P_1 \setminus \Pstop$. Otherwise, no transition $t \in T_c \subset T_1$ could remove the remaining tokens from any place $p \in P_1 \setminus \Pstop$, because $\inflow(p,t) = \bot$, and $\Mempty$ could not be reached.
For the collapsing sequence $M_n \goto{*} \Mempty$, there exists an equal collapsing sequence $M_n' \goto{*} \Mempty$ of transitions $t \in T_c \subset T_2$.
For every object $o_m \in \objects$ of type $ot(o_m) = \sigma_m$ and every relationship $(\sigma_m, \sigma_o)$, there exists a place $p_{\sigma_m, \sigma_o} \in P_L$ with exactly one token $(o_m, o_o) \in M_n'(p_{\sigma_m, \sigma_o})$. There exists a final place $p \in \Pstop$ with $o_m \in M_n'(p)$, and $F_{in,2}(p, t_e) = \tup{x_{\sigma_m}}$ and the link place with $F_{in,2}(p_{\sigma_m, \sigma_o}, t_e) = \tup{x_{\sigma_m}, x_{\sigma_o}}$. Hence, a binding $\beta'(\tup{x_{\sigma_m}, x_{\sigma_o}}) = \{o_m, o_o\}$ exists for every $o_m$ and is enabled.
Therefore, all objects can be removed with their many-to-one links.

A simple inductive argument using ($\star$), and the observations about the emitting and link creating run, and the collapsing run, shows shows that the total run \eqref{eq:Sync:run:ex} is possible.

\item[(3)] 
A run in $N_1$ is not conform to a stable relationship $(\sigma_m, \sigma_o) \in R$, if there exist objects $o_m, o_o \in \objects$, where $ot(o_m) = \sigma_m$ and $ot(o_o) = \sigma_o$, that contradict the condition $|lo(o_m, \sigma_o)| = 1$ and $lo(o_o, \sigma_m) \neq \emptyset$.
A contradiction for $N_1$ is posed by any of the following three options: (i) $lo(o_m, \sigma_o) = \emptyset$, (ii) $|lo(o_m, \sigma_o)| > 1$, or (iii) $lo(o_o, \sigma_m) = \emptyset$.

All transitions $t \in T_1 \subset T_2$ that operate both object types $\sigma_m$ and $\sigma_o$, also read from the link place for $(\sigma_m, \sigma_o)$ as $F_{in,2}(p_{\sigma_m, \sigma_o}, t) = F_{out,2}(t, p_{\sigma_m, \sigma_o}) = \tup{X_{\sigma_m}, x_{\sigma_o}}$.
Any binding $\beta_i$ for the transition $t$ operating on $\sigma_m$ and $\sigma_o$ is bound to the link place $p_{\sigma_m, \sigma_o}$. A binding $\beta_i$ including $o_m$ and $o_o$ is only enabled in the marking $M_i'$ if  $\tup{o_m, o_o} \in M_i'(p_{\sigma_m, \sigma_o})$.
$lo()$ is defined on the bindings $\beta_i$ of the run $M_0' \goto{*} M_n'$.

Therefore, we show that the links for $(\sigma_m, \sigma_o)$ can not be constructed in any way such that any of the contradictions (i-iii) are possible to hold in an accepted run of $N_2$.
Note that every object $o_m$ and $o_o$ is considered exactly once for a transition that creates links for the relationship $(\sigma_m, \sigma_o)$ in $N_2$ and during the run $M_0' \goto{*} M_n'$, all links in $p_{\sigma_m, \sigma_o}$ are stable (cf. item~(2)).
Hence, the following arguments hold true:
\begin{compactitem}
    \item[(i)] To satisfy $lo(o_m, \sigma_o) = \emptyset$, there must exist no binding that operates on $o_m$ and any $o_o$ of $ot(o_o) = \sigma_o$.
    For any accepted run, $o_m$ must participate in the link creation and there exists an object $o_o$ such that $p_{\sigma_m, \sigma_o}$ holds $\tup{o_m, o_o}$. A run in $N_2$ can only be accepted if any binding over $o_m$ and an object of type $\sigma_o$ binds $o_m$ to the appropriate $o_o$. Hence, $lo(o_m, \sigma_o) \neq \emptyset$
    \item[(ii)] To satisfy $|lo(o_m, \sigma_o)| > 1$, bindings must exist that correlate $o_m$ to two objects of type $\sigma_o$.
    However, during link creation, any $o_m$ can be linked to only one object of the type $ot(o_o) = \sigma_o$ such that in any marking only one link exists $| \{ \tup{o_m',o_o'} \in M_i'(p_{\sigma_m, \sigma_o}) \mid o_m = o_m' \} | = 1$.
    Any two bindings that consider $o_m$ and any object $o_o$ must read from the place $p_{\sigma_m, \sigma_o}$ and bind objects according to the links. Therefore, both must bind the same $o_o$. Hence, $|lo(o_m, \sigma_o)| \leq 1$.
    \item[(iii)] Analog to item (i), $lo(o_o, \sigma_m) = \emptyset$ is not possible as $o_o$ must participate in the link creation to be in an accepted run. As all bindings over $o_o$ and any objects of type $\sigma_m$ must adhere to the links. Hence $lo(o_o, \sigma_m) \neq \emptyset$ 
\end{compactitem}
By showing that all three contradictions must be false, the theorem holds true. No run \ref{eq:Sync:rej:ex} exists for $N_2$ that creates links for a non-conform execution.
\qed
\end{compactenum}
\end{proof}

\end{document}